\documentclass[a4paper,11pt]{article}
\usepackage{mathpazo}
\usepackage{multirow}
\usepackage{nicefrac}
\usepackage{setspace}
\usepackage[margin=1in]{geometry}
\usepackage{xcolor}
\definecolor{xlinkcolor}{cmyk}{1,1,0,0}
\usepackage{hyperref}
\hypersetup{colorlinks=true,linkcolor=black,citecolor=xlinkcolor}
\usepackage{color}
\usepackage{amsthm}
\usepackage{graphicx}
\usepackage{tablefootnote}
\usepackage{amsmath}
\usepackage{amssymb}
\usepackage{bbm}
\usepackage{dsfont}

\usepackage{diagbox}
\usepackage{makecell}
\usepackage{float}

\newtheorem{theorem}{Theorem}[section]

\newtheorem{corollary}[theorem]{Corollary}
\newtheorem{lemma}[theorem]{Lemma}

\newtheorem{claim}[theorem]{Claim}
\newtheorem{hypothesis}[theorem]{Hypothesis}
\newtheorem{proposition}[theorem]{Proposition}
\newtheorem{definition}[theorem]{Definition}

\newtheorem{remark}[theorem]{Remark}
\newtheorem{fact}[theorem]{Fact}
\newtheorem{open}[theorem]{Open Problem}

\usepackage{xspace}
\newcommand{\cost}{\mathsf{cost}}

\newcommand{\kmed}{$k$-$\mathsf{median}$\xspace}
\newcommand{\kmean}{$k$-$\mathsf{means}$\xspace}
\newcommand{\minsum}{$k$-$\mathsf{minsum}$\xspace}

\newcommand{\gapeth}{$\mathsf{Gap}$-$\mathsf{ETH}$\xspace}
\newcommand{\NP}{$\mathsf{NP}$\xspace}
\newcommand{\NPH}{$\mathsf{NP}\neq \mathsf{P}$}

\newcommand{\Good}{\mathsf{Good}}

\newcommand{\poly}{\text{poly}}

\newcommand{\con}{\mathsf{cov}}
\newcommand{\Po}{\mathcal{P}}
\newcommand{\C}{\mathcal{C}}

\newcommand{\tA}{\tilde{A}}

\renewcommand{\S}{\mathcal{S}}

\newcommand{\R}{\mathbb{R}}
\newcommand{\met}{\mathbb{R}_{\ge 1}}
\newcommand{\half}{\left(\vec{\frac{1}{2}}\right)}
\newcommand{\eps}{\varepsilon}

\renewcommand{\tilde}{\widetilde}

\newcommand{\NN}{\mathbb{N}}      

\newcommand{\AG}{\mathsf{AG}}
\newcommand{\sat}{\mathsf{sat}}
\newcommand{\wsat}{\mathsf{weak}\text{-}\mathsf{sat}}
\usepackage{enumerate}

\newcommand{\APX}{$\mathsf{APX}$\xspace}

\newcommand{\UGC}{$\mathsf{UGC}$}

\newcommand{\mkc}{Max $k$-Coverage\xspace}

\newcommand{\JCH}{$\mathsf{JCH}$\xspace}
\newcommand{\JCHD}{$\mathsf{JCH}^*$\xspace}
\newcommand{\Tu}{\mathsf{Tu}}

\usepackage{footnote}
\makesavenoteenv{tabular}
\makeatletter
\newcommand\footnoteref[1]{\protected@xdef\@thefnmark{\ref{#1}}\@footnotemark}
\makeatother

\newcommand\anote[2]{{\color{blue} **{\bf #1:} { #2}**}}

\newcommand\enote[1]{\anote{Euiwoong}{#1}}

\newcommand{\Turan}{Tur\'an\xspace}

\newcommand{\jc}{Johnson Coverage\xspace}
\newcommand{\jcd}{Johnson Coverage$^*$\xspace}

\newcommand{\hastad}{H\aa stad\xspace}
\newcommand{\lc}{Label Cover\xspace}

\newcommand{\call}{\mathcal{L}}
\newcommand{\calm}{\mathcal{M}}
\newcommand{\calh}{\mathcal{H}}
\newcommand{\cald}{\mathcal{D}}
\newcommand{\calp}{\mathcal{P}}
\newcommand{\Ex}{\mathbb{E}}
\newcommand{\wh}{\widehat}

\makeatletter
\renewcommand\paragraph{\@startsection{paragraph}{4}{\z@}%
                                      {\parskip}
                                      {-1em}%
                                      {\normalfont\normalsize\bfseries}}
\makeatother

\title{\vspace{-1.5cm}\textbf{
Johnson Coverage Hypothesis:\\ Inapproximability of   $k$-means and $k$-median in $\ell_p$-metrics}}
\date{}
\author{Vincent Cohen-Addad\thanks{
   Google Research, Switzerland.
   \texttt{vcohenad@gmail.com}.
  }
\and 
Karthik C.\ S.\footnote{Rutgers University, USA.
   \texttt{karthik.cs@rutgers.edu}. 
   }
\and 
Euiwoong Lee\footnote{ University of Michigan, USA.
   \texttt{euiwoong@umich.edu}. 
   }
}
\begin{document} 
\maketitle  \vspace{-0.5cm} \thispagestyle{empty} 

\begin{abstract}
 \small\baselineskip=9pt
  \kmed and \kmean are the two most popular objectives for clustering algorithms. Despite intensive effort, a good understanding of the  approximability of these
  objectives, particularly in $\ell_p$-metrics, remains a major open problem. In this paper, we significantly improve upon the hardness of approximation factors known in literature
  for these objectives in $\ell_p$-metrics.\vspace{0.1cm}

We introduce a new hypothesis called the \emph{Johnson Coverage Hypothesis} (\JCH), which roughly asserts that the well-studied \mkc problem on set systems is hard to approximate  to a factor greater than $\left(1-\nicefrac{1}{e}\right)$, even when the membership graph of the set system is a subgraph of the Johnson graph. We then show that together with generalizations of the embedding techniques introduced by Cohen-Addad and Karthik (FOCS '19),
  \JCH implies hardness of approximation results for \kmed and \kmean in $\ell_p$-metrics for factors which are close to the ones obtained for
  general metrics. 
  In particular, assuming \JCH we show that it is hard to
  approximate the \kmean objective:
  \begin{itemize}
  \item Discrete case: To a factor of 3.94 in the $\ell_1$-metric  and to a factor of 1.73 in the $\ell_2$-metric;
    this improves upon the previous factor of 1.56 and 1.17 respectively,  obtained under the Unique Games Conjecture (\UGC).
  \item Continuous case: To a factor of 2.10 in the $\ell_1$-metric  and to a factor of 1.36 in the $\ell_2$-metric;
    this improves upon the previous factor of   1.07 in the $\ell_2$-metric   obtained under  \UGC\ (and to the best of our knowledge, the continuous case of \kmean in $\ell_1$-metric was not previously analyzed in literature).
  \end{itemize}
  We also obtain similar improvements under \JCH for the \kmed objective.\vspace{0.1cm}
  
 Additionally, we prove a weak version of \JCH using the work of Dinur et al.\ (SICOMP~'05) on Hypergraph Vertex Cover, and recover all the results stated above of Cohen-Addad and Karthik  (FOCS '19) to (nearly) the same inapproximability factors but now under the standard \NPH\ assumption (instead of \UGC).\vspace{0.1cm}

Finally, we establish a strong connection between \JCH and the long standing open problem   of determining  the Hypergraph \Turan number.
We then use this connection to prove improved SDP gaps (over the existing factors in literature) for \kmean and \kmed objectives.

\end{abstract}


\clearpage \setcounter{page}{1}

\section{Introduction}
Over the last 15 years, the Unique Game Conjecture has enabled 
tremendous progress on the understanding of the
(in)approximability of fundamental graph problems, such as  Vertex Cover \cite{KR08}, Max Cut \cite{KKMO07},
Betweeness \cite{CGM09}, and Sparsest Cut \cite{CKKRS06,KV15,AKKSTV08}, sometimes even leading to tight results.
Yet, the approximability of most \NP-hard \emph{geometric} problems remains wide open: even
for fundamental  geometric problems such as the Traveling Salesman Problem, 
Steiner Tree, \kmed, and Facility Location, the hardness of approximation in various $\ell_p$-metrics  has remained below 1.01. On the other hand, approximation algorithms  for the aforementioned problems in $\Omega(\log n)$ dimensions do not achieve much better guarantees
than what they do for general metrics. 
This situation stands in stark contrast to the bounded dimension versions of the same problems,
whose approximabilities are mostly well understood.  

The main approach for proving hardness of approximation for geometric problems in \mbox{$\ell_p$-metrics} consists of two components: (1) establishing hardness of approximation in general metric spaces; 
(2) finding an embedding of the hard instances into $\ell_p$-metrics that preserves the gap.
The challenge here is to make (1) and (2) \emph{meet} at a sweet spot: we would like the hard gap instances in general metric spaces to be `embeddable', i.e., to be
mapped into the $\ell_p$-metric space while preserving the distances between points, but, proving a significant inapproximability bound for `embeddable' instances
requires a deep understanding of the hard instances (in the general metric) of the problem at hand.
This paper aims at characterizing the sweet spot between (1) and (2), hence providing a general framework for obtaining strong
inapproximability for a large family of clustering and covering problems.

\paragraph{Hardness of Clustering Problems in $\ell_p$-metrics.}
Given a set of points in a metric space, a clustering is a partition of the points such
that points in the same part are close to each other. Thus studying the complexity of finding good clustering is a very natural research avenue: On one hand these problems
have a wide range of applications, ranging from 
unsupervised learning, to information retrieval, and bioinformatics; on the other hand,
as we will illustrate in this paper, such problems are very natural generalizations of
set-cover-type problems to the metric setting, and as such are very fundamental computational problems.
The most popular objectives for clustering in metric spaces are arguably the \kmed and \kmean
problems: Given a set of points $P$ in a metric
space, the \kmed problem
asks to identify a set of $k$ representatives, called \emph{centers}, such that the sum of
the distances from each point to its closest center is minimized (for the \kmean problem, the goal
is to minimize the sum of distances squared) -- see Section~\ref{sec:prelims} for formal definitions.
In general metrics, the \kmed and \kmean problems are known to be hard to approximate within a factor
1.73 and  3.94 respectively~\cite{GuK99}, whereas the best known approximation algorithms achieve
an approximation guarantee of 2.67
and 9 respectively~\cite{BPRST15,ANSW16}.

A natural question arising from the above works is whether one can exploit the structure of more
specific metrics, such as doubling or Euclidean metrics to obtain better approximation or bypass the
lower bound.
If the points lie in an Euclidean space of arbitary dimension,
a 6.357-approximation and a 2.633-approximation are known for \kmean and \kmed respectively~\cite{ANSW16},
while a near-linear time approximation scheme is known for doubling metrics~\cite{abs-1812-08664}\footnote{With doubly exponential
  dependency in the dimension}.
In terms of hardness of approximation: the problems were known to be \APX-Hard
since the early 2000s~\cite{T00,GI03} in Euclidean spaces of dimension $\Omega(\log n)$
and have recently been shown to be hard to approximate within a factor of 1.17 and 1.06
for \kmean and \kmed respectively~\cite{CK19}. 
Perhaps surprisingly, for some other structured metrics such as Hamming or Edit distance, no better approximation
algorithms are known while the  hardness known for both Hamming and Edit distances  were 1.56 for \kmean problem and 1.14  for $k$-median problem.  
Summarizing, the gap between upper and lower bounds
for the Euclidean, Hamming, and Edit metrics, remains huge.

\paragraph{Technical Barriers.}
A well-known framework to obtain hardness of approximation results in the general metric for clustering objectives is through a straightforward reduction from the Max $k$-Coverage\footnote{Given a set system and $k \in \NN$, the problem asks to choose $k$ sets to maximize the size of their union.} or the Set
Cover problem.  We create a 'point' for each element of the universe and a 'candidate center', namely a location
where it is possible to place a center, for each set. Then, we define the distance between a point (corresponding to an element of the universe) and a candidate center (corresponding to a set) to be 1 if the set contains the element and 3 otherwise. 
This reduction due to Guha and Khuller~\cite{GuK99} yields lower bounds of $1+2/e$ and $1+8/e$ for the \kmed and \kmean problems, respectively, in general discrete
metric spaces. 

However, for the \kmed and \kmean objectives, the above strong hardness results for  Max $k$-Coverage produce instances that are 
impossible to embed in $\R^n$ without suffering a huge distortion in the distances, and thus would yield only a trivial gap
for \kmean or \kmed problems in $\ell_p$-metrics.
This issue has been faced for most other geometric problems as well, such as TSP or Steiner tree.
It is tempting to obtain more structure on the ``hard instances'' of set-cover-type problems with the large toolbox that
has been developed around the Unique Games Conjecture over the last 15 years for fundamental graph problems such as Vertex
Cover or Sparsest Cut. However, it seems that most of these hardness of approximation tools would not provide the adequate
structure on the set cover instances constructed, namely a structure that would make the whole instance easily ``embeddable'' into
an $\ell_p$-metric.

This illustrates the main difficulty in understanding  \kmed / \kmean in $\ell_p$-metrics;  
Guha-Khuller type reductions from general set systems with strong (perhaps optimal) gaps are not directly ``embeddable'',
but the current results use hardness for very restricted systems in a black-box way and thus cannot be extended to give strong hardness of approximation results.
Thus, we ask: 

\begin{center}
\textit{What structure on hard instances of Max $k$-Coverage problem would yield \\ meaningful hardness of \kmed and \kmean in $\ell_p$-metrics?}
\end{center}

\subsection{Our Results}
\label{sec:results}

This paper addresses the above question  in a unified manner, 
providing a general framework for obtaining strong inapproximability in $\ell_p$-metrics (see Table~\ref{table} for a summary of our results).
We segregate our results below in terms of conceptual and technical contributions and provide more details in the subsequent subsubsections. 

	Our main conceptual contribution is proposing the Johnson Coverage Hypothesis (\JCH) and identifying it as lying at the heart of the hard instances of \kmean and \kmed problems in $\ell_p$-metrics.
Intuitively, \JCH conjectures that the $(1-1/e)$-hardness of approximation for \mkc \cite{F98} holds even for set systems whose bipartite incidence graph (between sets and elements) is a subgraph of  the {\em Johnson graph} (see Section~\ref{sec:JCHintro} for the formal definition). 
Johnson graphs are well-studied objects in combinatorics that also admit nice geometric embedding properties.
More generally, we argue that 
identifying an intermediate mathematical property between geometry and combinatorics,
and
revisiting fundamental combinatorial optimization problems (\mkc in this paper) restricted to instances having that property 
is a fruitful avenue to understand approximability of high dimensional geometric optimization problems. 

\begin{table}[]
\begin{center}{\renewcommand{\arraystretch}{1.45}
\resizebox{\linewidth}{!}{
\begin{tabular}{||c|c|c|c|c|c|c|c|c||}
    \cline{1-6}
    Metric &  \thead{Discrete\\ \kmean} & \thead{Discrete\\ \kmed} & \thead{Continuous\\ \kmean}& \thead{Continuous\\ \kmed}& Assumption& \multicolumn{1}{c}{}  \\ [0.5ex]\hline\hline
  \multirow{3}{*}{$\ell_0$}  & 1.56  \cite{CK19} &  1.14 \cite{CK19} &1.21  \cite{CK19} &1.07 \cite{CK19}  &\UGC&{\footnotesize Previous}\\ 
  & \textbf{1.38}   &  \textbf{1.12} &\textbf{1.16}& \textbf{1.06}&\NPH&{\footnotesize This paper}\\
     & \textbf{3.94}   &  \textbf{1.73} &\textbf{2.10}& \textbf{1.36}&\JCH/\JCHD&{\footnotesize This paper}\\ \hline
  \multirow{3}{*}{$\ell_1$}  & 1.56  \cite{CK19} &  1.14 \cite{CK19} &--&1.07 \cite{CK19} & \UGC &{\footnotesize Previous}\\ 
   & \textbf{1.38}   &  \textbf{1.12} &\textbf{1.16}& \textbf{1.06}&\NPH &{\footnotesize This paper}\\
     & \textbf{3.94}   &  \textbf{1.73} &\textbf{2.10}&\textbf{1.36}&\JCH/\JCHD &{\footnotesize This paper}\\ \hline
  \multirow{3}{*}{$\ell_2$}  & 1.17  \cite{CK19} &  1.06 \cite{CK19} &1.07  \cite{CK19} &--  &\UGC &{\footnotesize Previous}\\ 
  & \textbf{1.17}   &  \textbf{1.07} &\textbf{1.06}& \textbf{1.015}&\NPH &{\footnotesize This paper}\\
     & \textbf{1.73}   &  \textbf{1.27} &\textbf{1.36}& \textbf{1.08}&\JCH/\JCHD&{\footnotesize This paper}\\ \hline
  \multirow{1}{*}{$\ell_\infty$}  & 3.94  \cite{CK19} &  1.73 \cite{CK19} & Open\tablefootnote{\label{noteopen}See Open Problem~\ref{open:ellinf}.}  & Open\textsuperscript{\getrefnumber{noteopen}}&\NPH&{\footnotesize Previous}\\ \hline
  \multirow{1}{*}{General}  & 3.94  \cite{GuK99} &  1.73 \cite{GuK99} &4  \cite{CKL21} &  2 \cite{CKL21}& \NPH&{\footnotesize Previous}\\\hline
\end{tabular}
}
}
\end{center}
\vspace{-0.1cm}
\caption{ \footnotesize{State-of-the-art inapproximability for \kmean and \kmed clustering objectives in various metric spaces for both the discrete and
  continuous versions of the problem. All the results for $\ell_p$-metrics stated inside the table are    when the input points are given in $O(\log n)$ dimensions. 
The NP-hardness reductions for the continuous objectives are randomized. }
}\label{table}
\end{table}


Our main technical contributions are two-fold. First, we generalize the embedding technique introduced in \cite{CK19} that gives a reduction from covering problems to clustering problems. The embedding technique in \cite{CK19} has two components: the first component is embedding the incidence graph of the complete graph into $\ell_p$-metrics and the second component is a dimension reduction technique developed by designing efficient protocols for the Vertex Cover problem in the communication model. We generalize both these components: we show how to embed the incidence graph of the complete \emph{hypergraph} into $\ell_p$-metrics and we provide a dimension reduction technique by designing efficient protocols for the \emph{Set Cover} problem in the communication model.
These generalization were  necessary to prove strong inapproximability results for clustering objectives as we needed to handle the less structured instances arising from \JCH and its variants (as opposed to \cite{CK19} whose reduction started from the more structured Vertex Coverage problem). 

Combining \JCH with the new embedding results above, we  deduce strong and perhaps even surprising inapproximability results: for example, we show that \kmean and \kmed problems are no easier in the $\ell_1$-metric/Hamming metric/Edit distance metric than in general metric spaces\footnote{It is interesting to note that approximate nearest-neighbor search, a problem closely related to clustering, is easier to solve in the $\ell_1$-metric  than general metrics.}, or obtain a much higher inapproximability for $\ell_2$-metric (see Table~\ref{table} and Section~\ref{sec:kmeankmed} for the whole list of results we obtain). 
Indeed, our framework accommodates hardness of {\em any} variant of \JCH defined in Section~\ref{subsec:gjc} in order to show new hardness of clustering problems;
for instance, the classic Set Cover hardness of Feige (or any APX-hardness of \mkc) can be plugged into our framework to produce \APX-Hardness for \kmean and \kmed in all $\ell_p$-metrics in both the  discrete and the continuous case (see Remarks~\ref{rem:APXdisc}~and~\ref{rem:APXcont} for precise details)

Our second   technical contribution is a new \NP-Hardness of approximation result for a (weaker) variant of \JCH corresponding to Hypergraph Vertex Coverage in $3$-uniform hypergraphs.
Consequently, with the above generalized embedding technique, we obtained \NP-Hardness results for clustering problems to factors
that nearly match and sometimes even improve upon the previous best hardness results (which were based on the Unique Games Conjecture); see Table~\ref{table} for the precise factors.
Our proof relies on (i) tighter analysis of the {\em Multilayered PCP} constructed by Dinur et al.~\cite{DGKR05}, followed by (ii) a \emph{densification} process to ensure that the resulting instance is dense enough to be used for clustering hardness.
As the previous best \NP-hardness results for clustering~\cite{CK19} rely on the hardness of Vertex Coverage~\cite{AKS11,M19,AS19} obtained by recent advances on the 2-to-2 Games Conjecture~\cite{KMS17,DKKMS18b,DKKMS18a,KMS18,BKS19,BK19}, we believe that further study on \JCH and its variants will lead to more interesting ideas in hardness of approximation that will be useful for understanding geometric optimization problems.  

\subsubsection{Johnson Coverage Hypothesis}\label{sec:JCHintro}

We introduce the Johnson Coverage Hypothesis (\JCH) which states that for every $\varepsilon>0$, there is some $z\in\NN$, such that given a parameter $k\in\NN$ and  a collection $E$ of $z$-sized subsets ($z$-sets henceforth) of $[n]$  as input, it is \NP-hard to distinguish between the following two cases:
\begin{description}
\item[Completeness:] There are $k$ subsets of $[n]$ each of cardinality $z-1$, say $S_1,\ldots ,S_k$, such that for every $T\in E$ there is some $i\in[k]$ such that $S_i\subset T$.
\item[Soundness:] For every $k$ subsets of $[n]$ each of cardinality $z-1$, say $S_1,\ldots ,S_k$, we have that there is  $E'\subseteq E$ such that 
\begin{itemize}
\item For every $T\in E'$ and every $i\in [k]$ we have $S_i\not\subset T$.
\item $|E'|\ge \left(\frac{1}{e}-\eps\right)\cdot |E|$.
\end{itemize}
\end{description}
We will often say that a set $S$ {\em covers} another $T$ if $S \subseteq T$. 

We refer to the gap problem described above in \JCH as the Johnson Coverage problem. Notice that when $z=2$, the Johnson Coverage problem is just the Vertex Coverage problem, for which \cite{AS19} have shown that a value close to 0.93 
is the correct inapproximability ratio assuming the Unique Games Conjecture. 
Therefore, for larger values of $z$, \JCH suggests that the Johnson Coverage problem (which is a generalization of the Vertex Coverage problem) gets harder to approximate and approaches the same inapproximability as \mkc. 

Introducing \JCH is the main conceptual contribution of this paper. 
Most previous hardness results for geometric optimization problems~\cite{T00, GI03, ACKS15, LSW17}
use hardness of graph problems in bounded degree graphs and embed them to $\ell_p$-metrics, where the bounded degree condition was crucial for the inapproximability ratio. 
Cohen-Addad and Karthik \cite{CK19}, at least in the context of clustering, were the first to show how to embed instances of Vertex Coverage problem (with no restriction on the degree) to obtain hard instances of \kmean and \kmed.  

In this work, we further study embeddability of coverage problems, and show that the embedding of \cite{CK19} can be further generalized to covering problems on hypergraphs, i.e., to the Johnson Coverage problem, which may be as hard as the general covering problem.  
Johnson Coverage problem is a purely combinatorial covering problem, but the implicit geometric structure of the problem allows us to seamlessly embed it to $\ell_p$-metrics.

\paragraph{Plausibility of \JCH and Connection to Hypergraph \Turan Number.} 
One may reformulate the Johnson Coverage problem as a special case of \mkc, by asking for instances of \mkc whose sets and elements correspond to subsets of size $(z-1)$ and $z$ of some universe $[n]$ respectively. So it is natural to study the performance of the standard LP or SDP relaxation, which is closely related to relaxations for the clustering problems. 

A natural gap instance for Johnson Coverage problem is the complete $z$-uniform hypergraph, i.e.,  $E:= \binom{[n]}{z}$. 
Since each $e \in E$ contains exactly $z$ subsets of size $(z-1)$, 
the standard LP relaxation and SDP relaxation admit a feasible solution whose value is at most $\nicefrac{\binom{[n]}{z - 1} }{z}$ and $\nicefrac{\binom{[n]}{z - 1} }{z - 1}$ respectively~\cite{ghazi2017lp}.
Interestingly, the soundness analysis of this proposed gap instance is closely related to the Hypergraph \Turan problem, a long standing open problem  in extremal combinatorics \cite{turan41research}. 

Given $r <~z \in \NN$, let $\Tu(n, z, r)$ be the minimum $f\in\NN$ such that there exists an $r$-uniform hypergraph $H=([n], F)$ with $n$ vertices and $f$ edges 
where every set $T \in \binom{[n]}{z}$ contains at least one hyperedge from $F$. 
Let $t(z, r) = \lim_{n \to \infty} \Tu(n, z, r) \binom{n}{r}^{-1}$. 
The classical Turan's theorem states that $t(z, 2) = \nicefrac{1}{(z - 1)}$, with the extremal example being the union of $(z-1)$ equal-sized cliques. 
However, the quantity of our interest $t(z, z - 1)$ is not well understood when $z > 3$; 
the best lower bound is $\nicefrac{1}{(z - 1)}$~\cite{de1983extension} and the best upper bound is $\nicefrac{(1/2 + o(1))\ln z}{z}$~\cite{sidorenko1997upper},
with the conjecture $t(z, z-1) \geq \omega(1/z)$ still open~\cite{de1991current}.
We refer the reader to more recent surveys by Sidorenko~\cite{sidorenko1995we} and Keevash~\cite{keevash2011hypergraph}. 

What will happen if $H$ has at most a $\nicefrac{1}{(z - 1)}$ fraction of hyperedges, which is much less than the conjectured value of $t(z, z-1)$? Then some $T \in \binom{[n]}{z}$ will not be covered by hyperedges of $H$. 
If we consider a random hypergraph where each hyperedge is picked with probability $\nicefrac{1}{z-1}$, $S \in \binom{[n]}{z}$ is covered with probability $1 - (1 - \nicefrac{1}{z-1})^z$, which converges to $1 - 1/e$ as $z$ increases. It is thus natural to hypothesize that this is indeed optimal as $z$ increases. 

\begin{hypothesis}
Any $(z-1)$-uniform hypergraph with $n$ vertices and $\nicefrac{\binom{n}{z - 1}}{z - 1}$ hyperedges covers at most a $d(n, z)$ fraction of sets of size $z$, where $\lim_{z \to \infty} \lim_{n \to \infty} d(n, z) = 1 - 1/e$. 
\label{hyp:extremal}
\end{hypothesis}

If Hypothesis~\ref{hyp:extremal} is true, then the afore-proposed gap instance for Johnson Coverage problem (the complete $z$-uniform hypergraph) has an integrality gap of $(1 - 1/e + \eps)$ for any $\eps > 0$ for the LP and SDP relaxation. 
Thus, a refutation of \JCH either implies interesting constructions in extremal combinatorics or establishes that there is a polynomial time algorithm outperforming the LP and SDP relaxation, on which the current best approximation algorithms for both \kmed and \kmean in any metric are based~\cite{BPRST15, ANSW16}.

At first glance, the hypothesis looks rather strong (i.e., less likely to be true) because of the existence of extremal structures that are strictly better than random hypergraphs. However, these advantages often diminish as the size of the object of interest increases. One example is the $z$-clique density in graphs corresponding to $t(z, 2)$. After the work of Razborov (for $z = 3$~\cite{razborov2008minimal}) and Nikiforov (for $z = 4$~\cite{nikiforov2011number}), Reiher~\cite{reiher2016clique} proved that the graph that covers the most number of $K_z$ with the given edge density is the union of disjoint cliques of the same size. 
In our context, since every $K_z$ has $\binom{z}{2}$ edges, the desired edge density is $\nicefrac{1}{w - 1}$ with $w = \binom{z}{2}$, so the extremal graph is the union of $w$ disjoint cliques and the probability that a random copy of $K_z$ is not covered by this graph is 
$
\prod_{i = 1}^z (1 - \frac{i - 1}{\binom{z}{2}-1}),
$
which converges to $1/e$ as $z$ increases. So for large $z$, the extremal examples do not have an advantage over a random graph!
Indeed, this connection already gives improved SDP gaps for various clustering objectives in natural hard instances in $\ell_p$-metrics (see Theorem~\ref{theorem:gaps:informal}, or for more details, see Section~\ref{section:gaps}).

\paragraph{Technical Barriers.} One may wonder if it is possible to start from Feige's hard instances of \mkc \cite{F98} (which have the desired gap of $1-\nicefrac{1}{e}$) and simply provide a gap-preserving reduction to Johnson Coverage problem. In Theorem~\ref{thm:fpt} we showed that if such a reduction exists, then it should significantly blow up the witness size (the number of sets needed to cover the universe in the completeness case). At the heart of this observation is that we are transforming arbitrary set systems (arising from hard instances of \mkc) to set systems with bounded VC dimension (as in the Johnson Coverage problem). In fact, a reduction in this spirit was recently obtained  in \cite{CKL21}, and that reduction did indeed blow up the witness size. Therefore the result in \cite{CKL21} where we show \mkc is hard to approximate beyond $1-\nicefrac{1}{e}$ factor on set systems of large girth (thus bounded VC dimension) may be seen as moral progress on understanding \JCH. 

Additionally, in Theorem~\ref{thm:jchgadget}, we revisit Feige's framework for showing hardness of approximation of \mkc, and highlight that certain simple modifications to his reduction would not prove \JCH. In particular, we show that no ``partition system'' can be combined with standard label cover instances to yield \JCH. This provides some evidence that in order to prove \JCH, we would need to potentially prove hardness of approximation result for some  highly structured label cover instances.

\paragraph{Subsequent Work.} 
Motivated by an earlier version of this paper, Guruswami and Sandeep \cite{GS20} initiated the study of the minimization variant of \JCH, 
where the goal is to select as few $(z-1)$-sets as possible to ensure that every $z$-set in $E$ is covered by a chosen $(z-1)$-set. 
While the naive algorithm gives a $z$-approximation, they obtained an $(z/2 + o(z))$-approximation algorithm. 
They also asked as an open problem whether \JCH and hardness of the minimization version can be formally related.

\subsubsection{Generalized Johnson Coverage Problem}
\label{subsec:gjc}
Motivated by the  connections in previous subsubsection to $t(z, r)$ for general $z > r \geq 1$, we consider more generalized Johnson Coverage problems 
where the input is still a collection of $z$-sized sets $E$ but instead of choosing $(z - 1)$-sized sets, we choose $r$-sized sets for some $r \in \{1,\ldots ,z-1\}$ to cover as many sets in $E$ as possible. (See Definition~\ref{def:gjc} for the formal definition.) For example, when $r = 1$, the problem becomes the $z$-Hypergraph Vertex Coverage problem. We prove the following for the $3$-Hypergraph Vertex Coverage problem:

\begin{theorem}[Informal statement; See Theorem~\ref{thm:np-hard}]\label{thm:informalNP}
For every $\varepsilon>0$, the $3$-Hypergraph Vertex Coverage problem is \NP-Hard to approximate to a factor of $7/8+\varepsilon$.
\end{theorem}

While the  covering version of the above problem, namely the {\em (Minimum) $3$-Hypergraph Vertex Cover} has been studied actively in literature, culminating in~\cite{DGKR05}, the coverage version does not seem to have been explicitly studied. 

We remark that the above inapproximability result is higher than the inapproximability of Vertex Coverage problem proved under \UGC\ (of roughly 0.93) and more so under \NPH\ (of roughly 0.98), and this higher gap will be useful in the next subsubsection for applications.

In Section~\ref{sec:cond}, we provide a embedding framework which converts any hardness result for Generalized Johnson Coverage problem (where $z > r \geq 1$) with hardness ratio $\alpha < 1$  
to directly yield inapproximability results for various clustering objectives. 
Together with Theorem~\ref{thm:informalNP}, 
this framework produces \NP-Hardness results nearly matching or improving the best known hardness results assuming  \UGC. 
The case $z = 4$, $r = 2$ also yields improved SDP gaps for various clustering objectives in Section~\ref{section:gaps}.

\subsubsection{Inapproximability Results for \kmean and \kmed in $\ell_p$-metrics}\label{sec:kmeankmed}
First,  we present our results for the ``discrete'' \kmed and \kmean
problems. In these versions, the centers must be chosen from a specific
set of points of the metric.

\begin{theorem}[Discrete version; Informal statement of Theorems~\ref{thm:kmeanellp} and \ref{thm:kmedellp}]\label{thm:introellp} 
Assuming \JCH,  given $n$ points and $\poly(n)$ candidate centers,  it is \NP-hard to approximate:
\begin{itemize}
  \item the \kmean objective to within a $1~+~\nicefrac{8}{e}\approx 3.94$ factor in $\ell_1$-metric and  
    $1~+~\nicefrac{2}{e}\approx 1.73$ factor in $\ell_2$-metric.
  \item the \kmed objective to within a $1~+~\nicefrac{2}{e}\approx 1.73$ factor in $\ell_1$-metric and  
    $1.27$ factor in $\ell_2$-metric.
    \end{itemize}
\label{thm:discrete-informal}
\end{theorem}

In fact, we obtain inapproximability results for all $\ell_p$-metrics (where $p\in\met$), and the details are provided in Section~\ref{sec:cond}.
Also, the results obtained in the above theorem significantly improves on the bounds of \cite{CK19} (see Table~\ref{table}).
Finally, we note that the bounds obtained for the $\ell_1$-metric might be optimal as  approximating \kmean and \kmed to a factor of $3.95$ and $1.74$ is \emph{fixed parameter tractable} even for general metrics \cite{Cohen-AddadG0LL19}. 

We now sketch the proof of Theorem~\ref{thm:introellp}.  We use the notations concerning \JCH established in Section~\ref{sec:JCHintro}. We create a point for each ($z$ size) subset in $E$ and a candidate center for each subset of $[n]$ of size $z-1$. Both the points and candidate centers are just the characteristic vectors of their corresponding subsets of $[n]$, i.e., the points are all Boolean vectors of Hamming weight $z$ and the candidate centers are all Boolean vectors of Hamming weight $z-1$. In the completeness case, it is easy to see that there is a set $\mathcal{C}$ of $k$ candidate centers such that every point has a center in $\mathcal{C}$ at (Hamming) distance 1. Also, in the soundness case, it is easy to see that  for every set $\mathcal{C}$ of $k$ candidate centers, we have that  at least $\nicefrac{1}{e}-\varepsilon$ fraction of the points in $E$ are at (Hamming) distance at least 3  from every center in $\mathcal{C}$ . Note that the dimension of this embedding is $n$, and that we  reduce it to $O(\log n)$ by developing a   generalization of the dimenionality reduction machinery  introduced in \cite{CK19}. Furthermore, the proof from the Hamming metric to other $\ell_p$-metrics goes  through the composition of various graph embedding gadgets introduced in \cite{CK19} and generalized in this paper to hypergraph embedding. 


We now shift our attention to  the continuous case, where the centers can be picked at arbitrary locations in the metric space.  We show the following\footnote{Note that Theorem~\ref{thm:introell2without} depends on a slight strengthening of \JCH, which we call \JCHD, where we further assume that for the hard instances of \JCH, we have $|E|=\omega(k)$ (see Section~\ref{sec:condCont} for further discussion).
}.

\begin{theorem}[Continuous version; Informal statement of Theorems~\ref{thm:kmeanl2},~\ref{thm:kmedl1},~\ref{thm:kmedl2},~and~\ref{thm:kmeanl1}]\label{thm:introell2without}
Assuming \JCHD,  given $n$ points,  it is \NP-hard to approximate:
\begin{itemize}
  \item the \kmean objective to within $ 2.10$ factor in $\ell_1$-metric and  
    $1~+~\nicefrac{1}{e}\approx 1.36$ factor in $\ell_2$-metric.
  \item the \kmed objective to within a $1~+~\nicefrac{1}{e}\approx 1.36$ factor in $\ell_1$-metric and  
    $1.08$ factor in $\ell_2$-metric.
    \end{itemize}
\label{thm:continuous-informal}
\end{theorem}

The above theorem for \kmed in $\ell_1$-metric is obtained through an intermediate hardness of approximation
proof for \kmed in the Hamming metric (see Theorem~\ref{thm:kmedHam}).
Additionally, the fact that  medians and means have nice algebraic definitions in $\ell_1$ and $\ell_2$ respectively
allow us to transfer the hardness from Hamming metrics without much loss.
Furthermore, thanks to a technical result of Rubinstein \cite{R18} on near isometric embedding from Hamming metric in $d$ dimensions to Edit metric in $O(d\log d)$ dimensions, we can translate all our (discrete case) results in the Hamming metric to the Edit metric
(see Appendix~A in \cite{CK19} for details).

We also prove the first hardness of approximation results for \kmed in $\ell_2$-metric (Theorem~\ref{thm:kmedl2}) and \kmean in $\ell_1$-metric (Theorem~\ref{thm:kmeanl1}). 
Even though continuous \kmed in $\ell_2$ has been actively studied for the bounded $d$ or $k$~\cite{arora1998approximation, badoiu2002approximate, kumar2010linear, feldman2011unified}, and to the best of our knowledge, no hardness of approximation for the general case was known in the literature. We remark that independent to our work, in \cite{BGJ21}, the authors prove \APX-hardness of the Euclidean \kmed problem under \UGC.

Next, we move our attention to proving \NP-Hardness results. 
With general versions of Theorem~\ref{thm:discrete-informal} and Theorem~\ref{thm:continuous-informal} for generalized Johnson Coverage problem, 
Theorem~\ref{thm:informalNP} proves the following \NP-Hardness results presented in Table~\ref{table}.
For continuous versions, more technical work and randomized reductions are needed to ensure enough density (see Section~\ref{subsec:dense}).

\begin{theorem}
It is \NP-hard to approximate the following clustering objectives; 
discrete \kmean in $\ell_1$ within $1.38$, 
discrete \kmed in $\ell_1$ within $1.12$, 
discrete \kmean in $\ell_2$ within $1.17$, 
discrete \kmed in $\ell_2$ within $1.07$, 
continuous \kmean in $\ell_1$ within $1.16$, 
continuous \kmed in $\ell_1$ within $1.06$, 
continuous \kmean in $\ell_2$ within $1.06$, 
and continuous \kmed in $\ell_2$ within $1.015$.
The results for continuous \kmed and \kmean hold under randomized reductions. 
\label{thm:np-informal}
\end{theorem}

Finally, we present another evidence for usefulness of (generalized) Johnson Coverage problem
via the clique density theorem of Reiher~\cite{reiher2016clique}, 
which gives improved SDP gaps of $149/125 = 1.192$ for discrete \kmed in $\ell_1$-metric  and \kmean in $\ell_2$-metric
when the instances are {\em well-separated}. (I.e., points are not closed to one another.)
Moreover, this result holds  even when the integral solution is allowed to use $\Omega(k)$ more centers. 
Previously, the best gaps were $1.14$ and $1.17$ for \kmed in $\ell_1$-metric and \kmean in $\ell_2$-metric respectively, 
both following from the SDP gaps of the Unique Games hardness~\cite{CK19}.

\begin{theorem} [Informal version of Theorem~\ref{theorem:gaps}]
Fix any $\eps > 0$. 
For discrete \kmed in $\ell_1$ and discrete \kmean in $\ell_2$, there is a family of well-separated instances where
the SDP relaxation has a gap of at least $149/125 - \eps \approx 1.192 - \eps$, even when the integral solution opens $\Omega(k)$ more centers.
\label{theorem:gaps:informal}
\end{theorem}

\subsection{Organization of the Paper}
The paper is organized as follows. In Section~\ref{sec:prelims}, we introduce some notations that are used throughout the paper, and some tools from coding theory. 
In Section~\ref{sec:cond}, we introduce \JCH, and show how it implies the inapproximability of \kmean and \kmed in various $\ell_p$-metrics (i.e., Theorems~\ref{thm:introellp}~and~\ref{thm:introell2without}). 
In Section~\ref{sec:nphard} we prove a weak version of \JCH. 
In Section~\ref{sec:open} we present some open problems of interest.

\section{Preliminaries}
\label{sec:prelims}

\paragraph*{Notations.} For any two points $a,b\in\mathbb{R}^d$, the distance between them in the $\ell_p$-metric is denoted by $\|a-b\|_p~=~\left(\sum_{i=1}^d|a_i-b_i|^p\right)^{1/p}$.
Their distance in the $\ell_{\infty}$-metric is denoted by 
$\|a-b\|_{\infty} = \underset{{i\in[d]}}{\max}\ \{|a_i-b_i|\}$, and
in the $\ell_0$-metric is denoted by 
$\|a-b\|_0=|\{i\in[d]:a_i\neq b_i\}|$, i.e.,
the number of coordinates on which $a$ and $b$ differ. For every $n\in\mathbb{N}$, we denote by $[n]$ the set of first $n$ natural numbers, i.e., $\{1,\ldots,n\}$. We denote by $\binom{[n]}{r}$, the set of all subsets of $[n]$ of size $r$. Let $e_i$ denote the vector which is 1 on coordinate $i$ and 0 everywhere else. We denote by $\half$, the vector that is $\nicefrac{1}{2}$ on all coordinates.

\paragraph*{Clustering Objectives.} Given two sets of points $P$ and $C$ in a metric space, we define
 the \kmean cost of $P$ for $C$
 to be $\underset{p \in P}{\sum}\left(\underset{c \in C}{\min}\ \left(\text{dist}(p,c)\right)^2\right)$ and
 the \kmed cost to be $\underset{p \in P}{\sum}\left(\underset{c \in C}{\min}\ \text{dist}(p,c)\right)$.
 Given a set of points $P$  in a metric space and partition $\pi$ of $P$ into $P_1\dot\cup P_2\dot\cup\cdots \dot\cup P_k$, we define
 the \minsum cost of $P$ for $\pi$
 to be $\underset{i\in [k] }{\sum}\left(\underset{p,q \in P_i}{\sum}\ \text{dist}(p,q)\right)$.
 Given a set of points $P$, the \kmean/\kmed (resp.\ \minsum)
 objective is the minimum over all $C$ (resp.\ $\pi$) of cardinality $k$ of the \kmean/\kmed (resp.\ \minsum)
 cost of $P$ for $C$ (resp.\ $\pi$). Given a point $p \in P$, the contribution to the
 \kmean (resp.\ \kmed) cost of $p$ is $\underset{c \in C}{\min}
 \left(\text{dist}(p,c)\right)^2$ (resp.\ $\underset{c \in C}{\min}\ \text{dist}(p,c)$).

\paragraph*{Error Correcting Codes.}
We recall here a few coding theoretic notations. 
An error correcting code of block length $\ell$ over alphabet set $\Sigma$ is simply a collection of codewords $\C \subseteq \Sigma^\ell$.  The relative distance between any two points is the fraction of coordinates on which they are different. The relative distance of the code $\C$  is defined  to be the smallest relative distance between any pair of distinct codewords in $\C$. The message length of $\C$ is defined to be $\log_{|\Sigma|} |\C|$. The rate of $\C$ is defined as the ratio of its message length and block length.  

\begin{theorem}[\cite{GS96,SAKSD01}]\label{thm:code}
For every prime square $q$ greater than 49, there is a code family denoted by $\AG$ over alphabet of size $q$ of positive constant (depending on $q$) rate  and relative distance at least $1-\frac{3}{\sqrt{q}}$. Moreover, the encoding time of any code in the family is polynomial in the message length. 
\end{theorem}

An informal argument justifying the existence of the above code family is provided in \cite{CK19}. Furthermore, as noted in \cite{CK19}, random codes obtaining weaker parameters than the parameters stated above (see Gilbert-Varshamov bound \cite{G52,V57}) suffice for the results in this paper and it may even be possible to use concatenated codes (arising from Reed-Solomon codes) which approach the Gilbert-Varshamov bound in the proofs in this paper instead of the aforementioned algebraic geometric codes.

\section{Conditional Inapproximability of \kmean and \kmed in $\ell_p$-metrics}\label{sec:cond}

In this section, we first  formally introduce the Johnson Coverage Hypothesis (\JCH{}), then generalize the gadget constructions via graph embedding which were introduced in \cite{CK19}, and finally prove how \JCH implies the hardness of approximation results for \kmean and \kmed, i.e., Theorems~\ref{thm:introellp}~and~\ref{thm:introell2without}.
We also prove improved integrality gaps inspired by \JCH (i.e., Theorem~\ref{theorem:gaps:informal}).

\subsection{Johnson Coverage Hypothesis}\label{sec:JCH}

In this subsection, we first introduce the \jc problem, followed by the \jc hypothesis.

Let $n,z,y\in\mathbb N$ such that $n\ge z > y$. Let $E\subseteq \binom{[n]}{z}$ and $S\in \binom{[n]}{y}$. We define the coverage of $S$ w.r.t.\ $E$, denoted by $\con(S,E)$ as follows:
$$
\con(S,E)=\{T\in E\mid S\subset T\}.
$$

\begin{definition}[Johnson Coverage Problem]
In the $(\alpha,z,y)$-Johnson Coverage problem with $z > y \geq 1$, we are given a  universe $U:=[n]$, a collection of subsets of $U$, denoted by $E\subseteq \binom{[n]}{z}$, and a parameter $k$ as input. We would like to distinguish between the following two cases:
\begin{itemize}
\item \textbf{\emph{Completeness}}: There exists $\C:=\{S_1,\ldots ,S_k\}\subseteq  \binom{[n]}{y}$ such that 
$$\con(\C):=\underset{i\in[k]}{\cup} \con(S_i,E)=E.$$
\item \textbf{\emph{Soundness}}: For every $\C:=\{S_1,\ldots ,S_k\}\subseteq \binom{[n]}{y}$ we have ${\left|\con(\C)\right|}\le \alpha\cdot\left|E\right|$.
\end{itemize}
\label{def:gjc}
We call $(\alpha, z, z-1)$-Johnson Coverage as $(\alpha, z)$-Johnson Coverage. 
\end{definition}

Notice that $(\alpha,2)$-Johnson Coverage Problem is simply the well-studied vertex coverage problem (with gap $\alpha$). Also, notice that if instead of picking the collection $\C$ from $\binom{[n]}{y}$, we replace it with picking the collection $\C$ from $\binom{[n]}{1}$ with a similar notion of coverage, then we simply obtain the  Hypergraph Vertex Coverage problem (which is equivalent to the \mkc problem for unbounded $z$).

We now put forward the following hypothesis. 

\begin{hypothesis}[Johnson Coverage Hypothesis (\JCH)]
For every constant $\varepsilon>0$, there exists a constant $z:=z(\varepsilon)\in\mathbb{N}$ such that deciding the $\left(1-\frac{1}{e}+\varepsilon,z\right)$-Johnson Coverage Problem is \NP-Hard. 
\end{hypothesis}

Since Vertex Coverage problem is a special case of the \jc problem, we have that the \NP-Hardness of $(\alpha,z)$-Johnson Coverage problem is already known for $\alpha=0.94$ \cite{AS19} (under unique games conjecture). 

On the other hand,  if we replace picking the collection $\C$ from $\binom{[n]}{z-1}$ by picking from $\binom{[n]}{1}$, then for the Hypergraph Vertex Coverage problem, we do know that for every $\eps>0$ there is some constant $z$ such that the Hypergraph Vertex Coverage problem is \NP-Hard to decide for a factor of $\left(1-\frac{1}{e}+\varepsilon\right)$ \cite{F98}.

\paragraph{Related work to \JCH.}
Johnson Coverage problem can be considered as a special case of \mkc in set systems with a natural additional structure. 
Such a restriction of fundamental covering / packing / constraint satisfaction problems to structured instances often arise in geometric settings such as 
Independent Set of Rectangles~\cite{chuzhoy2016approximating}. 
The {\em VC dimension} has been an important combinatorial notion capturing many implicit structures posed by geometric problems that allow better algorithms~\cite{bronnimann1995almost, badanidiyuru2012approximating} than general set systems in some regimes. 
Recently, Alev et al.~\cite{alev2019approximating} studied approximating constraint satisfaction problems on a high-dimensional expander, 
another object that bridges graph theory and geometry, and showed that this structure makes CSPs easier to approximate.

\subsection{Gadget Constructions via Graph Embeddings}\label{sec:gadget}
In this subsection, we recall a notion of graph embedding that was introduced in \cite{CK19}. And then we prove some bounds on the embedding for important $\ell_p$-metrics. These embedding are then used in the next subsection to prove inapproximability results. 

Let $q,t,r\in\mathbb N$ such that $q\ge t\ge r$.  Let $J(q,t,r)$ denote the incidence graph of the Johnson graph \cite{HS93}.
 Elaborating,  we define  $J(q,t,r)$ to be the bipartite graph on partite sets $\binom{[q]}{t}$ and $\binom{[q]}{r}$ where we have an edge $(S,S')\in \binom{[q]}{t}\times \binom{[q]}{r}$ in $J(q,t,r)$  if and only if $S'\subseteq S$. 
We use the shorthand $J(q,t)$ to denote $J(q,t,t-1)$.

We would like to analyze the embedding of $J(q,t)$ into $\ell_p$-metric spaces for all $p\in\met$. 

\begin{definition}[Gap Realization of a Bipartite graph \cite{CK19}]
  \label{def:bipartitegadget}
Let $p\in\met$. For any bipartite graph $G = (A\dot\cup B, E)$ and $\lambda \geq 1$, a mapping $\tau: V \to \mathbb{R}^d$ is said to \emph{$\lambda$-gap-realize $G$ (in the $\ell_p$-metric)} if for some $\beta > 0$, the following holds:
\begin{enumerate}[(i)]
\item For all $(u, v) \in E$, $\|\tau(u) - \tau(v)\|_p = \beta$.
\item For all $(u, v) \in (A \times B) \setminus E$, we have $\|\tau(u) - \tau(v)\|_p \ge \lambda\cdot \beta$.
\end{enumerate}
Moreover, we require that $\tau$ $\lambda$-gap-realize $G$ in the $\ell_p$-metric efficiently, i.e., there is a polynomial time algorithm (in the size of $G$) which can compute $\tau$. 
\end{definition}

We remark here that the above definition is a variant of the notion gap contact dimension introduced in \cite{KM19} and is closely related to notion of contact dimension  which has been well-studied in literature since the early eighties \cite{P80,M85,FM88,M91,DKL19}. We refer the reader to \cite{CK19} for further discussion.

\begin{definition}[Gap number]
Let $p\in\met$. For any bipartite graph $G = (A\dot\cup B, E)$, its gap number in the $\ell_p$-metric $g_p(G)$ is the largest $\lambda$ for which  there exists a mapping $\tau$ that $\lambda$-gap-realizes $G$ in a $d$-dimensional $\ell_p$-metric space\footnote{For all the main results of this paper to hold,  we do not require the specified upper bound on the dimension of the mapping realizing the gap number; any finite dimensional realization suffices.} where $d\le \poly(|A|+|B|)$.
\end{definition}

In \cite{CK19}, the authors studied $g_p(J(q,t,1))$. In this paper, we are interested in analyzing $g_p(J(q,t,t-1))$ for all $q,t\in\mathbb{N}$ ($q\ge t$) and $p\in \met$. We remark that we do not study $g_{\infty}(J(q,t,t-1))$ in this paper, as this was already settled to be equal to 3 in \cite{CK19}. We recall the following upper bound on gap number which follows immediately from triangle inequality:

\begin{proposition}[Essentially \cite{CK19}]\label{prop:ubcd}
Let $q\ge 3$, $t\ge 2$ (where $q\ge t$),  and $p\in\met$. We have $g_p(J(q,t))\le 3$. 
\end{proposition}

We consider the $\ell_1$-metric and show that we can meet the upper bound in Proposition~\ref{prop:ubcd}.
\begin{lemma}\label{lem:cd1}
For all $q\ge 3$ and $t\ge 2$ (where $q\ge t$), we have $g_1(J(q,t))=3$. 
More generally, $g_1(J(q,t,s))=(t - s + 2) / (t - s)$.
\end{lemma}
\begin{proof}
For the $\ell_1$-metric consider the mapping $\tau: \binom{[q]}{t}\cup \binom{[q]}{s}\to\{0,1\}^q$ defined as follows. For every  $S\in \binom{[q]}{t}\cup \binom{[q]}{s}$, we define $$\tau(S)=\sum_{i\in S}e_i.$$ 

Fix some $(S,S')\in \binom{[q]}{t}\times \binom{[q]}{s}$ such that $S'\subseteq S$. Then we have that $$\tau(S)-\tau(S')=\sum_{i\in S\setminus S'}e_i\Rightarrow \|\tau(S)-\tau(S')\|_1=t - s.$$ 

On the other hand if we fix some $(S,S')\in \binom{[q]}{t}\times \binom{[q]}{s}$ such that $S'\not\subset S$ then we have that 
$$
\|\tau(S)-\tau(S')\|_1= \left\|\left(\sum_{i\in S\setminus S'}e_i\right)-\left(\sum_{i\in S'\setminus S}e_i\right)\right\|_1\ge t - s + 2.
$$ 

Thus we have that $\tau$, $(t-s+2)/(t-s)$-gap-realizes $J(q,t,s)$ in the $\ell_1$-metric.  
\end{proof}

Now we focus our attention to bounding the gap number in the Euclidean metric. First, we see that the below  lower bound simply follows from $\tau$ constructed in the above proof. 

\begin{corollary}\label{cor:cd2} 
For all $q\ge 3$ and $t\ge 2$ (where $q\ge t$), we have $g_2(J(q,t))\ge \sqrt{3}$
and $g_2(J(q,t))\ge \sqrt{(t-s+2)/(t-s)}$. 
\end{corollary}

However, we improve the lower bound with a different embedding. 

\begin{lemma}\label{lem:cd2}
For all $q\ge 3$ and $t\ge 2$ (where $q\ge t$), we have $g_2(J(q,t,1))\ge\sqrt{1+\frac{1}{\sqrt{t}-1}}$. 
More generally, $g_2(J(q,t,s))\ge\sqrt{1+\frac{1}{\sqrt{ts}-s}}$.
\end{lemma}
\begin{proof}
For the $\ell_2$-metric consider the mapping $\tau: \binom{[q]}{t}\cup \binom{[q]}{s}\to\mathbb{R}^q$ defined as follows. For every  $T\in \binom{[q]}{t}$, we define $$\tau(T)=\sum_{i\in T}e_i.$$ 

For every  $S\in  \binom{[q]}{s}$, we define $$\tau(S)=\sqrt{\frac{t}{s}}\cdot \sum_{i\in S}e_i.$$ 

Fix some $(T,S)\in \binom{[q]}{t}\times \binom{[q]}{s}$ such that $S\subseteq T$. Then we have that 
\allowdisplaybreaks\begin{align*}
\tau(T)-\tau(S)&=\left(\left(\sqrt{\frac{t}{s}}-1\right)\cdot \sum_{i\in S}e_i\right)+\left(\sum_{i\in T\setminus S}e_i\right)\\
\Rightarrow \|\tau(T)-\tau(S)\|_2&=\sqrt{s\cdot\left(\sqrt{\frac{t}{s}}-1\right)^2+t-s}\\
&=\sqrt{2}\cdot \sqrt{t-\sqrt{ts}}
\end{align*}

On the other hand if we fix some $(T,S)\in \binom{[q]}{t}\times \binom{[q]}{s}$ such that $S\not\subset T$ then we have that 
\begin{align*}
\|\tau(T)-\tau(S)\|_2&=\left\|\left(\left(\sqrt{\frac{t}{s}}-1\right)\cdot \sum_{i\in S\cap T}e_i\right)+\left(\sqrt{\frac{t}{s}}\cdot \sum_{i\in S\setminus T}e_i\right)+\left(\sum_{i\in T\setminus S}e_i\right)\right\|_2\\
&\ge \sqrt{(s-1)\cdot\left(\sqrt{\frac{t}{s}}-1\right)^2+t-s+1+\frac{t}{s}}\\
&=\sqrt{2}\cdot\sqrt{t-\sqrt{ts}+\sqrt{\frac{t}{s}}}\\
&=\sqrt{2}\cdot\sqrt{(t-\sqrt{ts})\cdot \left(1+\frac{1}{\sqrt{ts}-s}\right)}
\end{align*}

Thus we have that $\tau$, $\sqrt{\left(1+\frac{1}{\sqrt{ts}-s}\right)}$-gap-realizes $J(q,t,s)$ in the $\ell_2$-metric.  When $s=1$, $\tau$  $\sqrt{\left(1+\frac{1}{\sqrt{t}-1}\right)}$-gap-realizes $J(q,t,1)$ in the $\ell_2$-metric.
\end{proof}

In order to see that the lower bound on $g_2(J(q,t,s))$ given in Lemma~\ref{lem:cd2} is indeed higher than the one given in Corollary~\ref{cor:cd2}, note the following:
\begin{align*}
\left(1+\frac{1}{\sqrt{ts}-s}\right)-\left(\frac{t-s+2}{t-s}\right)&=\frac{1}{\sqrt{ts}-s}-\frac{2}{t-s}=\frac{t-s-2\sqrt{ts}+2s}{(\sqrt{ts}-s)(t-s)}=\frac{(\sqrt{t}-\sqrt{s})^2}{(\sqrt{ts}-s)(t-s)}> 0,
\end{align*}
as $t>s$.

We wrap up our computation of gap numbers by showing that as $p$ grows the gap number of $J(q,t)$ in the $\ell_p$-metric approaches 3.

\begin{lemma}\label{lem:cdp}
For all $q\ge 3$ and $t\ge 2$ (where $q\ge t$), we have  that for every $\varepsilon>0$ there exists $p\in\mathbb{N}$ such that $g_p(J(q,t))> 3-\varepsilon$. 
\end{lemma}
\begin{proof}Fix $q\ge 3$, $t\ge 2$, and $\varepsilon>0$. Let $p\in\mathbb{N}$ such that $q^{1/p}<1+\nicefrac{\varepsilon}{3}$. 
Consider the mapping $\tau: \binom{[q]}{t}\cup \binom{[q]}{t-1}\to\mathbb R^q$ defined as follows. For every  $S\in  \binom{[q]}{t-1}$, we define $$\tau(S)=\half+\sum_{i\in S}e_i,$$ 
and for every  $T\in \binom{[q]}{t}$, we define $$\tau(T)=\sum_{i\in T}e_i.$$

Fix some $(S,T)\in \binom{[q]}{t-1}\times\binom{[q]}{t}$ such that $S\subset T$. Then we have that $$\eta:=\tau(T)-\tau(S)=\left(\sum_{i\in T\setminus S}e_i\right)-\half.$$ Since $\eta\in\{\nicefrac{-1}{2},\nicefrac{1}{2}\}^q$, we have that $\|\eta\|_p=\|\tau(T)-\tau(S)\|_p=\nicefrac{q^{1/p}}{2}$.  

On the other hand if we fix some $(S,T)\in \binom{[q]}{t-1}\times\binom{[q]}{t}$ such that $S\not\subset T$, i.e., $\exists u\in [q]$ such that $u\in S\setminus T$  then we have that 
$$
\|\tau(T)-\tau(S)\|_p\ge |(\tau(T))_u-(\tau(S))_u|=\frac{3}{2}.
$$ 

Thus we have that $\tau$, $\left(\frac{3}{q^{1/p}}\right)$-gap-realizes $J(q,t)$ in the $\ell_p$-metric.   Finally note that $\frac{3}{q^{1/p}}>\frac{9}{3+\varepsilon}=3-\frac{3\varepsilon}{3+\varepsilon}>3-\varepsilon $.
\end{proof}

We remark that the embeddings given in Lemmas \ref{lem:cd1} and \ref{lem:cdp} are essentially small modifications of the ones given in \cite{CK19}.

For the sake of compactness of statements in the future, we introduce the following.
\begin{definition}
For all $p\in\met$, we define $\gamma_{p}=\underset{\substack{q\ge t}}{\inf}\ g_p(J(q,t))$, 
and $\gamma_{p, \Delta}=\underset{\substack{q\ge t \geq \Delta + 1}}{\inf}\ g_p(J(q,t,t - \Delta))$
\end{definition}

Finally, we conclude this subsection by recalling the following `hereditary' property of the embedding which is important for applications to hardness of approximation.

\begin{proposition}[\cite{CK19}]\label{prop:gadget}
Let $q\ge 3$ and $t\ge 2$ (where $q\ge t$), and $p\in\met$. Let $E\subseteq \binom{[q]}{t}$ and $G$ be the subgraph of $J(q,t)$ induced by the vertex set $E\cup\binom{[q]}{t-1}$. Let $\tau$ be a $\lambda$-gap realization of $J(q,t)$ in the $\ell_p$-metric.  Then $\tau$  restricted to the vertices of $G$ is a $\lambda$-gap realization of $G$ in the $\ell_p$-metric. 
\end{proposition}

\subsection{Conditional Inapproximability of Discrete \kmean and \kmed}\label{sec:condDisc}

In this subsection, we show how \JCH or any hardness of $(\alpha,z,y)$-\jc implies the hardness of approximation of the discrete case of \kmean and \kmed in all $\ell_p$-metrics, i.e., we prove Theorem~\ref{thm:introellp}.

First, we define for every $p\in\met$,the quantities $\zeta_{1}(p, \Delta, \alpha)$ and $\zeta_{2}(p, \Delta, \alpha)$ as follows: $$\zeta_{1}(p, \Delta, \alpha):=1+ (1 - \alpha)(\gamma_{p, \Delta}-1) \text{\ \ and \ \ }\zeta_{2}(p, \Delta, \alpha):=1+ (1 - \alpha) (\gamma_{p, \Delta}^2-1).$$ 

Also let $\zeta_1(p) := \zeta_{1}(p, 1, 1-1/e)$ and $\zeta_2(p) := \zeta_{2}(p, 1, 1-1/e)$. 
Notice that $\zeta_1(1)\approx 1.73$, $\zeta_2(1)=3.94$,  $\zeta_1(2)\approx 1.27$, $\zeta_2(2)\approx 1.73$, and as $p\rightarrow\infty$, we have $\zeta_1(p)\rightarrow\zeta_1(1)$ and $ \zeta_2(p)\rightarrow \zeta_2(1)$. 

Now, we state our inapproximability results for \kmean and \kmed.

\begin{theorem}[\kmean with candidate centers in $O(\log n)$ dimensional $\ell_p$-metric space]\label{thm:kmeanellp} Let $p\in\met$.
Assuming 
$(\alpha, z, y)$-Johnson Coverage is NP-hard,
for every constant $\varepsilon>0$,  given a point-set $\Po\subset \mathbb{R}^{d}$ of size $n$ (and $d=O(\log n)$), a collection $\C$ of $m$ candidate centers in $\mathbb R^d$ (where $m=\poly(n)$), and a parameter $k$ as input, it is \NP-Hard to distinguish between the following two cases:
\begin{itemize}
\item \textbf{\emph{Completeness}}:  
There exists $\C':=\{c_1,\ldots ,c_k\}\subseteq \C$ and $\sigma:\Po\to\C'$ such that $$\sum_{a\in\Po}\|a-\sigma(a)\|_p^2\le \rho n (\log n)^{2/p},$$
\item \textbf{\emph{Soundness}}: For every $\C':=\{c_1,\ldots ,c_k\}\subseteq \C$ and every $\sigma:\Po\to\C'$ we have: $$\sum_{a\in\Po}\|a-\sigma(a)\|_p^2\ge (\zeta_2(p, z -y, \alpha)-\varepsilon)\cdot \rho n (\log n)^{2/p},$$
\end{itemize}
for some constant $\rho>0$.
\end{theorem}

In particular, assuming \JCH, \kmean is \NP-Hard to approximate within a factor $\zeta_2(p)$, which is at least $3.94$ in $\ell_1$ and at least $1.73$ in $\ell_2$. 

\begin{theorem}[\kmed with candidate centers in $O(\log n)$ dimensional $\ell_p$-metric space]\label{thm:kmedellp}
Let $p\in\met$.
Assuming 
$(\alpha, z, y)$-Johnson Coverage is \NP-Hard, 
for every constant $\varepsilon>0$,  given a point-set $\Po\subset \mathbb{R}^{d}$ of size $n$ (and $d=O(\log n)$), a collection $\C$ of $m$ candidate centers in $\mathbb R^d$ (where $m=\poly(n)$), and a parameter $k$ as input, it is \NP-Hard to distinguish between the following two cases:
\begin{itemize}
\item \textbf{\emph{Completeness}}:  
There exists $\C':=\{c_1,\ldots ,c_k\}\subseteq \C$ and $\sigma:\Po\to\C'$ such that $$\sum_{a\in\Po}\|a-\sigma(a)\|_p\le \rho n(\log n)^{1/p},$$
\item \textbf{\emph{Soundness}}: For every $\C':=\{c_1,\ldots ,c_k\}\subseteq \C$ and every $\sigma:\Po\to\C'$ we have: $$\sum_{a\in\Po}\|a-\sigma(a)\|_p\ge (\zeta_1(p, z-y, \alpha)-\varepsilon)\cdot \rho n(\log n)^{1/p},$$
\end{itemize}
for some constant $\rho>0$.
\end{theorem}
In particular, assuming \JCH, \kmed is \NP-Hard to approximate within a factor $\zeta_2(p)$, which is at least $1.73$ in $\ell_1$ and at least $1.27$ in $\ell_2$. 

\begin{remark}\label{rem:APXdisc}
Notice that for every $q,z,$ and $p$, from the embedding given by Lemma~\ref{lem:cd1}, we can deduce that  $g_p(J(q,z,1))$ is a constant (depending on $q,z,$ and $p$) bounded away from 1. Thus we have  for $i\in\{1,2\}$ that $\zeta_i(p,z-1,\alpha)\ge 1+\delta$ for some positive constant $\delta$  depending only on $p,\alpha,$  and $z$.  On the other hand, Feige \cite{F98} showed that for any constant $\varepsilon>0$ there is a constant $z:=z(\varepsilon)$ such that $(1-\nicefrac{1}{e}+\varepsilon,z,1)$-Johnson Coverage problem is \NP-Hard. Plugging these parameters in Theorems~\ref{thm:kmeanellp}~and~\ref{thm:kmedellp}, we have that \kmean and \kmed are \APX-Hard in all $\ell_p$-metrics. In fact, we do not even need to invoke the tight result of \cite{F98}, as it's relatively easy to show that $(1-\delta,z,1)$-Johnson Coverage problem is \NP-Hard directly from the PCP theorem \cite{AS98,ALMSS98,D07} for some constants $\delta>0$ and $z\in\NN$.   \end{remark}

The proof of the above two theorems follows by generalizing the embedding framework established in \cite{CK19}. In \cite{CK19}, the authors considered a two-player communication game between Alice and Bob, where Alice holds as input an edge in a publicly known graph  and Bob holds as input a vertex in the same graph, and the goal is for Bob to send a message to Alice (using public randomness), so that Alice can decide if her edge is incident on Bob's vertex. The authors then showed how an efficient randomized protocol for this communication problem can be combined with a certain embedding of the complete graph (into $\ell_p$-metrics) to obtain inapproximaility results for \kmean and \kmed in $\ell_p$-metrics. 

Below we consider the game where Alice holds  a $z$-sized subset of $[n]$  and Bob holds a $y$-sized subset of $[n]$, and the goal is for Bob to send a message to Alice (using public randomness), so that Alice can decide if her subset completely contains Bob's subset. We then design an efficient randomized protocol for this communication problem using Algebraic Geometric codes (as in \cite{CK19}) and show how the transcript of the protocol can be combined with a certain  embedding of the Johnson graph (into $\ell_p$-metrics) to obtain the two theorems. We skip providing further details about the framework for the sake of brevity, and point the reader to 
\cite{CK19}. 

\begin{proof}[Proof of Theorems \ref{thm:kmeanellp} and \ref{thm:kmedellp}]
Fix $\varepsilon>0$ as in the theorem statement. Also, fix $p\in\met$. Let $\eps':=\eps/11$ (setting $\eps'$ to be $\eps/4$ suffices for Theorem~\ref{thm:kmedellp}). Starting from a hard instance of $\left( \alpha, z, y \right)$-\jc problem $(U,E,k)$,
  we create an instance of the \kmean, or of the \kmed problem using  Algebraic-Geometric codes and the embedding given in Proposition~\ref{prop:gadget}  as follows. Let $q$ be a (constant) prime square greater than $\left(\frac{18z^2}{\eps'}\right)^2$. Let $\AG$ be the code guaranteed by Theorem~\ref{thm:code}  over alphabet of size $q$ of message length $\eta:=\log_q n$ (recall $|U|=n$), block length $\ell:=O_{q}(\eta)$, and relative distance at least $1-\nicefrac{3}{\sqrt{q}}$. 
  Let $\tau$ be a $g_p(J(q,z,y))$-gap realization embedding of $J(q,z,y)$ which maps vertices of $J(q,z,y)$ to a $d^*$-dimensional space (note that $q$, $z$, and $y$ are all constants, and thus so is $d^*$). Let  $\beta>0$ be the constant from Definition~\ref{def:bipartitegadget}.

\paragraph{Construction.}  The \kmed or \kmean instance consists of the set of candidate centers $\C\subseteq \R^{\ell\cdot d^*}$  
  and the set of points to be clustered $\Po\subseteq \R^{\ell\cdot d^*}$ which will be defined below. 
 First, we define functions $A_E:E\times [\ell]\to \binom{[q]}{z}\cup \{\perp\}$ and $A_F:\binom{[n]}{y}\times [\ell]\to \binom{[q]}{y}\cup \{\perp\} $ below. Then, we will construct functions $\tilde A_E:E\to\mathbb{R}^{d^*\cdot \ell}$ and $\tilde A_F:\binom{[n]}{y}\to\mathbb{R}^{d^*\cdot \ell}$.  Given $\tilde A_E$ and $\tilde A_F$ the point-set $\Po$ is just defined to be $$\left\{\tilde A_E(T)\big| T\in E\right\},$$ and the set of candidate centers $\C$ is just defined to be $$\left\{\tilde A_F(S)\Big| S\in\binom{[n]}{y}\right\}.$$ 

For every $\gamma\in [\ell]$ and every $S\in\binom{[n]}{y}$, we define $R_{S,\gamma}^{y}\subseteq [q]$, where $\mu\in [q]$ is contained in $R_{S,\gamma}^{y}$ if and only if there exists some $u\in S$ such that\footnote{Here we think of $\gamma$ as a field element of $\mathbb F_q$ by using some canonical bijection between $[q]$ and $\mathbb F_q$.} $\AG(u)_\gamma=\mu$. Then, we define $A_F(S,\gamma)=R_{S,\gamma}^{y}$ if $|R_{S,\gamma}^{y}|=y$ and $A_F(S,\gamma)=\perp$ otherwise. Similarly, for every $\gamma\in [\ell]$ and every $T\in E\subseteq \binom{[n]}{z}$  we define $R_{T,\gamma}^z\subseteq [q]$, where $\mu\in [q]$ is contained in $R_{T,\gamma}$ if and only if  there exists some $u\in T$ such that $\AG(u)_\gamma=\mu$. Then, we define $A_E(T,\gamma)=R_{T,\gamma}^{z}$ if $|R_{T,\gamma}^{z}|=z$ and $A_E(T,\gamma)=\perp$ otherwise.

For every $x,y\in[q]$ such that $x<x'$ we define $\Lambda_{x,x'}:\binom{[q]}{x}\to \binom{[q]}{x'}$ as follows: For every $X\in \binom{[q]}{x}$ define $\Delta_X$ to be the set of $(x'-x)$ many smallest integers in $[q]$ not contained in $X$. Then $\Lambda_{x,x'}(X)=X\cup \Delta_X$.  Now we can construct functions $\tilde A_E:E\to\mathbb{R}^{d^*\cdot \ell}$ and $\tilde A_F:\binom{[n]}{x'}\to\mathbb{R}^{d^*\cdot \ell}$ as follows:
$$\forall \gamma\in[\ell], \ 
\tilde A_E(T)\lvert_{\gamma}=\begin{cases} \tau(A_E(T,\gamma)) \text{ if } A_E(T,\gamma)\neq \perp\\ \tau(\Lambda_{|R_{T,\gamma}^z|,z}(R_{T,\gamma}^{z}))\text{ otherwise}\end{cases}$$
$$ \text{ and }\ \ \ \tilde A_F(S)\lvert_{\gamma}=\begin{cases} \tau(A_F(S,\gamma)) \text{ if } A_F(S,\gamma)\neq \perp\\ \tau(\Lambda_{|R_{S,\gamma}^{x'}|,x'}(R_{S,\gamma}^{x'}))\text{ otherwise}\end{cases}.
$$

\paragraph{Structural Observations.}
Fix $S\in \binom{[n]}{y}$. We define the set $L_S^{y}$ as follows: $$L_S^{y}:=\{\gamma\in[\ell]:|R_{S,\gamma}^{y}|=y\}.$$ 
Consider the set $W_S=\{\AG(u)\mid u\in S\}$. Since any two codewords of $\AG$ agree on at most $3/\sqrt{q}$ fraction of coordinates, we have by union bound that there are at least $1-\binom{y}{2}(3/\sqrt{q})$ fraction of coordinates of $[\ell]$ on which all codewords in $W_S$ are distinct. Therefore, we have that $|L_S^{y}|\ge (1-\binom{y}{2}(3/\sqrt{q}))\cdot \ell$. Also, note that for all  $\gamma\in L_S^{y}$ we have $A_F(S,\gamma)\neq\perp$.

Similarly, we fix $T\in E$. We define the set $L_T^{z}$ as follows: $$L_T^{z}:=\{\gamma\in[\ell]:|R_{T,\gamma}^{z}|=z\}.$$ 
By following the averaging argument above, we also have that  $|L_T^{z}|\ge (1-\binom{z}{2}(3/\sqrt{q}))\cdot \ell$. Again, note that for all  $\gamma\in L_T^{z}$ we have $A_E(T,\gamma)\neq\perp$.

 Finally, for every $(S,T)\in \binom{[n]}{y}\times E$, we define $L_{S,T}:=L_{S}^{y}\cap L_T^z$, and note the following 
 $$|L_{S,T}|\ge \left(1-\frac{3zy}{\sqrt{q}}\right)\cdot \ell.$$ 
  
Let us now compute a few distances.   
Consider  $(S,T)\in \binom{[n]}{y}\times E$ such that $S\subset T$ then we have
\begin{align}\label{eqgood}
 \|\tilde{A}_E(T)-\tilde{A}_F(S)\|_p=\beta\cdot \ell^{1/p}
.\end{align}
This is because, if $S\subset T$ then for all $\gamma\in[\ell]$ we have $R_{S,\gamma}^{y}\subseteq R_{T,\gamma}^z$. Fix $\gamma\in [\ell]$. If $\gamma\in L_{S,T}$ then we have that $ \|\tilde{A}_E(T)|_\gamma-\tilde{A}_F(S)|_\gamma\|_p^p=\beta^p$ as $(R_{S,\gamma}^{y},R_{T,\gamma}^z)$ is an edge in $J(q,z)$. On the other hand if $\gamma\notin L_{T}^z$ then either $\gamma\notin L_{S}^{y}$ and we have that $\Lambda_{|R_{S,\gamma}^{y}|,y}(R_{S,\gamma}^{y})\subseteq \Lambda_{|R_{T,\gamma}^z|,z}(R_{T,\gamma}^z)$ because $R_{S,\gamma}^{y}\subseteq R_{T,\gamma}^z$, or we have $\gamma\in L_{S}^{y}$, in which case, we have $R_{T,\gamma}^z=R_{S,\gamma}^{y}$ and thus $R_{S,\gamma}^{y}\subseteq \Lambda_{|R_{T,\gamma}^z|,z}(R_{T,\gamma}^z)$. In the former case, we have that $(\Lambda_{|R_{S,\gamma}^{y}|,y}(R_{S,\gamma}^{y}),\Lambda_{|R_{T,\gamma}^z|,z}(R_{T,\gamma}^z))$ is an edge in $J(q,z)$ and in the latter case we have that $(R_{S,\gamma}^{y},\Lambda_{|R_{T,\gamma}^z|,z}(R_{T,\gamma}^z))$ is an edge in $J(q,z)$. The final result in both cases is that $ \|\tilde{A}_E(T)|_\gamma-\tilde{A}_F(S)|_\gamma\|_p^p=\beta^p$. Therefore, we have  $$\sum_{\gamma\in[\ell]}\|\tilde{A}_E(T)|_\gamma-\tilde{A}_F(S)|_\gamma\|_p^p=\|\tilde{A}_E(T)-\tilde{A}_F(S)\|_p^p=\beta^p\cdot \ell.$$

Consider  $(S,T)\in \binom{[n]}{y}\times E$ such that $S\not\subset T$. Let $u\in U$ such that $u\in S$ but $u\notin T$. We define $\Good:=\{\gamma\in[\ell]: \AG(u)_\gamma\notin R_{T,\gamma}^z\}$. Since the relative distance of $\AG$ is $1-\nicefrac{3}{\sqrt{q}}$, we have by simple union bound that $$|\Good|\ge \left(1-\frac{3z}{\sqrt{q}}\right)\cdot\ell.$$  Then, we have
\begin{align}\label{eqbad}
\|\tilde{A}_E(T)-\tilde{A}_F(S)\|_p\ge (g_p(J(q,y,z))-\delta)\cdot \beta\cdot \ell^{1/p}
,\end{align}
where $\delta:=\frac{18zy}{\sqrt{q}}$ (note that $\delta\le \eps'$). To see this first observe that:

$$\|\tilde{A}_E(T)-\tilde{A}_F(S)\|_p^p\ge \sum_{\gamma\in L_{S,T}\cap \Good}\|\tilde{A}_E(T)|_\gamma-\tilde{A}_F(S)|_\gamma\|_p^p.$$ 
Fix $\gamma\in L_{S,T}\cap \Good$. Notice that $R_{S,\gamma}^{y}\not\subset R_{T,\gamma}^z$ and thus $(R_{S,\gamma}^{y},R_{T,\gamma}^z)$ is not an edge in $J(q,z)$. Therefore, we have
$$\|\tilde{A}_E(T)|_\gamma-\tilde{A}_F(S)|_\gamma\|_p^p\ge g_p(J(q,z,y))^p\cdot \beta^p.$$ 
Since $|L_{S,T}\cap \Good|\ge \left(1-\nicefrac{3z}{\sqrt{q}}-\nicefrac{3zy}{\sqrt{q}}\right)\cdot \ell\ge \left(1-\nicefrac{6zy}{\sqrt{q}}\right)\cdot \ell$, \eqref{eqbad} follows immediately by noting that (i) $g_p(J(q,z,y))\le 3$, and (ii) for any $\eta\in[0,1]$, $(1-\eta)^{1/p}\ge (1-\eta)$.
  
  We now analyze the \kmean and \kmed cost of the
  instance.
  Consider the completeness case first.

  \paragraph*{Completeness.}
  Suppose there exist $S_1,\ldots ,S_k\in \binom{[n]}{y}$ such that $\underset{i\in[k]}{\bigcup}\con(S_i,E)=E$. Then, we define $\C'=\{\tilde A_F(S_i)\mid i\in[k]\}$ and we define $\sigma:\Po\to\C'$ as follows: for every $a\in\Po$, where $a:=\tilde A_E(T)$ for some $T\in E$, let $\sigma(a)$ be equal to $\tilde A_F(S_i)$ such that $S_i\subset T$ (if there is more than one $i\in[k]$ for which $S_i$ is contained in $T$ then we choose one arbitrarily). Therefore for any $a
  \in\Po$  we have from \eqref{eqgood}
\begin{align*}
\|a-\sigma(a)\|_p=\beta\cdot \ell^{1/p}.
\end{align*}

  The \kmean cost of the overall instance is thus $ \beta^2\cdot\ell^{2/p}\cdot |\Po|$, while the \kmed cost is
  $\beta\cdot \ell^{1/p}\cdot|\Po|$. 
   Finally, we turn to the soundness analysis. 
  
    \paragraph*{Soundness.} Consider any set of centers $\C' = \{c_1,\ldots,c_k\}\subset \C$
  that is optimal for the \kmed or \kmean objective (and that $\sigma:\Po\to\C'$ simply maps to the closest center in $\C'$).
  Let $\S:=\{S_1,\ldots,S_k\}\subseteq \binom{[n]}{y}$ be the collection of $y$ sized subsets of $U$ corresponding to the centers of $\C'$, namely 
  $$
  \S=\left\{S_i\in \binom{[n]}{y} \Big| \tilde A_F(S_i)\in \C'\right\}.
  $$
  By the assumptions of the soundness case, we have $\con(\S)\le \alpha \cdot |E|$. For
  each $T \in\con(\S)$, from \eqref{eqgood}, we have that the contribution of $\tilde A_E(T)$ to the \kmean
  cost is exactly $\beta^2\cdot\ell^{2/p}$, and to the \kmed cost is exactly $\beta\cdot\ell^{1/p}$.
  However, from \eqref{eqbad}, for any other $T\in E\setminus \con(\S)$, the contribution of $\tilde A_E(T)$ to the \kmed and \kmean cost is (at least)
  respectively $(g_p(J(q,y,z))-\delta)\cdot \beta\cdot \ell^{1/p}$ and $(g_p(J(q,y,z))-\delta)^2\cdot \beta^2\cdot \ell^{2/p}$. 
  Therefore, the optimal solution w.r.t. \kmed objective has cost at least:
  \begin{align*}
  &\alpha \cdot |E|\cdot \beta\cdot \ell^{1/p}+\left(1 - \alpha \right)\cdot |E|\cdot (\gamma_{p, z-y}-\delta)\cdot \beta\cdot \ell^{1/p}\\
  &\ge |\Po|\cdot \beta\cdot \ell^{1/p}\cdot\left(\zeta_1(p, z-y, \alpha)-\delta\right)\\
  &\ge |\Po|\cdot \beta\cdot \ell^{1/p}\cdot\left(\zeta_1(p, z-y, \alpha)-\eps\right).   
  \end{align*}
  Similarly, the optimal solution w.r.t. \kmean objective has cost at least:
  \begin{align*}
  &\alpha \cdot |E|\cdot \beta^2\cdot \ell^{2/p}+\left(1 - \alpha \right)\cdot |E|\cdot (\gamma_{p,z-y}-\delta)^2\cdot \beta^2\cdot \ell^{2/p}\\
  &\ge |\Po|\cdot \beta^2\cdot \ell^{2/p}\cdot\left(\zeta_2(p, z-y, \alpha)-2\delta\right) \\ 
  &\ge |\Po|\cdot \beta^2\cdot \ell^{2/p}\cdot\left(\zeta_2(p, z-y, \alpha)-\eps\right).  
  \end{align*}
  \end{proof}

We remark that the above proof can be made significantly notation-light, if we wanted to prove Theorem~\ref{thm:introellp} in high dimensions (i.e., $\poly(n)$ dimensions). In particular, we could skip the use of Algebraic Geometric codes. However, the result is more interesting when the dimension is $O(\log n)$, as proving the same \NP-Hardness for $(\log n)^{1-o(1)}$ dimensions would violate the Exponential Time Hypothesis \cite{IP01,IPZ01}, due to the sub-exponential time approximation algorithm of \cite{Cohen-Addad18} for such dimensions.

\subsection{Conditional Inapproximability of Continuous \kmean and \kmed}\label{sec:condCont}


In this subsection, we show how \JCH or any hardness of $(\alpha,z,y)$-\jc
implies the hardness of approximation of various clustering objectives in the continuous case and i.e., we prove Theorem~\ref{thm:introell2without}.
In order to prove the results in this section, we need to assume a slight strengthening of \JCH. 
\begin{hypothesis}[Dense Johnson Coverage Hypothesis (\JCHD)]
\JCH holds for instances $(U,E,k)$ of \jc problem where $|E|=\omega(k)$.
\end{hypothesis}

More generally, let $(\alpha, z, y)$-\jcd be the special case 
$(\alpha, z, y)$-\jc where the instances satisfy $|E| = \omega(k \cdot |U|^{z - y - 1})$. 
Then \JCHD states that for any $\eps > 0$, there exists $z = z(\eps)$ such that 
$(1-1/e + \eps, z, z-1)$-\jcd is \NP-Hard. 
This additional property has always been obtained in literature by looking at the hard instances that were constructed. 
In \cite{CK19}, where the authors proved the previous best inapproximability results for continuous case \kmean and \kmed, 
it was observed that hard instances of $(0.94,2,1)$-\jc constructed in~\cite{AKS11, AS19} can be made to satisfy the above property.

First, we consider Euclidean \kmean.  

\begin{theorem}[\kmean without candidate centers in $O(\log n)$ dimensional Euclidean space]\label{thm:kmeanl2}
Assume $(\alpha, z, y)$-\jcd is \NP-Hard. 
For every constant $\varepsilon>0$,  given a point-set $\Po\subset \{0,1\}^{d}$ of size $n$ (and $d=O(\log n)$) and a parameter $k$ as input, it is \NP-Hard to distinguish between the following two cases:
\begin{itemize}
\item \textbf{\emph{Completeness}}:  
There exists $\C':=\{c_1,\ldots ,c_k\}\subseteq \mathbb R^d$ and $\sigma:\Po\to\C'$ such that $$\sum_{a\in\Po}\|a-\sigma(a)\|_2^2\le \rho n\log n,$$
\item \textbf{\emph{Soundness}}: For every $\C':=\{c_1,\ldots ,c_k\}\subseteq \C$ and every $\sigma:\Po\to\C'$ we have: $$\sum_{a\in\Po}\|a-\sigma(a)\|_2^2\ge 
\left(1 + \frac{(1 - \alpha)}{(z - y)} - \varepsilon\right)\cdot \rho n\log n,$$
\end{itemize}
for some constant $\rho>0$.
\end{theorem}
\begin{proof}
The construction of the point-set $\Po$ in the theorem statement is the same as in the proof of Theorem~\ref{thm:kmeanellp} (with a further simplification as we fix the embedding given in Lemma~\ref{lem:cd2}). However, the soundness analysis to prove the theorem is more intricate. 

Fix $\varepsilon>0$ as in the theorem statement. Also, fix $p\in\met$. Let $\eps':=\eps/11$. Starting from a hard instance of $\left( \alpha, z, y \right)$-\jcd problem $(U,E,k)$ with $|U| = n$ and $|E| = \omega(n^{z-y})$, 
we create an instance of the \kmean using  Algebraic-Geometric codes and the embedding given in Lemma~\ref{lem:cd2}  as follows. Let $q$ be a (constant) prime square greater than $\left(\frac{18z^2}{\eps'}\right)^2$. Let $\AG$ be the code guaranteed by Theorem~\ref{thm:code}  over alphabet of size $q$ of message length $\eta:=\log_q n$ (recall $|U|=n$), block length $\ell:=O_{q}(\eta)$, and relative distance at least $1-\nicefrac{3}{\sqrt{q}}$.

\paragraph{Construction.}  The  \kmean instance consists of the set of points to be clustered $\Po\subseteq \{0,1\}^{\ell\cdot q}$ which will be defined below. 
 Recall the function $A_E:E\times [\ell]\to \binom{[q]}{z}\cup \{\perp\}$ defined in the proof of Theorem~\ref{thm:kmeanellp} (also recall the definitions  $R_{T,\gamma}^{z}$, $\Lambda_{x,y}$, $L_T^{z},\tilde A_F$). We construct $\tilde A_E:E\to\{0,1\}^{q\cdot \ell}$ below.  Given $\tilde A_E$ the point-set $\Po$ is just defined to be $$\left\{\tilde A_E(T)\big| T\in E\right\}.$$

Let $\tau$ map a subset of $[q]$ to its characteristic vector in $\{0,1\}^q$. Now we can construct $\tilde A_E$ as follows:
$$\forall \gamma\in[\ell], \ 
\tilde A_E(T)\lvert_{\gamma}=\begin{cases} \tau(A_E(T,\gamma)) \text{ if } A_E(T,\gamma)\neq \perp\\ \tau(\Lambda_{|R_{T,\gamma}^z|,z}(R_{T,\gamma}^{z}))\text{ otherwise}\end{cases}.$$

\paragraph{Structural Observations.}
We  compute distances in a couple of cases.   
Consider  distinct $T,T'\in E$ such that $|T\cap T'|=y$. Then we have
\begin{align}\label{eqsub}
(2(z-y)-\delta)\cdot{\ell}  \le \|\tilde{A}_E(T)-\tilde{A}_E(T')\|_2^2\le {2(z-y)\ell}
,\end{align}
where $\delta:= 18z^2/\sqrt{q}$ (note that $\delta\le \eps'$).

Now consider  distinct $T,T'\in E$ such that $|T\cap T'|<y$. Then we have 
\begin{align}\label{eqnosub}
(2(z -y) + 2 -\delta)\cdot{\ell}\le \|\tilde{A}_E(T)-\tilde{A}_E(T')\|_2^2.
\end{align}
  
We now analyze the \kmean  cost of the
instance.
To do so, we will make use of the following classic fact about the $k$-means cost of a partition $C_1\dot\cup\cdots\dot\cup C_k=\Po$ of a set of
points in $\R^d$.
\begin{fact}
  \label{fact:kmeanscost}
  Given a clustering $C_1\dot\cup\cdots\dot\cup C_k=\Po$, the \kmean cost
  is exactly
  $$\sum_{i=1}^k \frac{1}{2|C_i|} \sum_{p \in C_i} \sum_{q \in C_i} \|p-q\|^2_2.
  $$
\end{fact}

Consider the completeness case first. 
\paragraph*{Completeness.}
  Suppose there exist $S_1,\ldots ,S_k\in \binom{[n]}{y}$ such that $\underset{i\in[k]}{\bigcup}\con(S_i,E)=E$. Then, we define a clustering $C_1\dot\cup\cdots\dot\cup C_k=\Po$  as follows: for every $a\in\Po$, where $a:=\tilde A_E(T)$ for some $T\in E$, let $a\in C_i$  such that $T\in \con(S_i,E)$ (if there is more than one $i\in[k]$ for which $S_i$ is contained in $T$ then we choose one arbitrarily).
We now provide an upper bound on the $k$-means cost of clustering $\mathcal{C}:=\{C_1,\ldots ,C_k\}$.
~\eqref{eqsub} implies that for each $C_i$, for any pair $T, T'$ such that
$\tA_E(T), \tA_E(T') \in C_i$, we have that $||\tA_E(T) - \tA_E(T') ||^2_2 \le 2\ell$.
Hence, if we let $W_i$ be the \kmean cost for the $C_i$, 
$$W_i = \frac{1}{2|C_i|} \sum_{q \in C_i} \sum_{p \in C_i} ||p-q||^2_2 \le
\frac{1}{2|C_i|} \sum_{q \in C_i} \sum_{p \in C_i, p \neq q} 2(z - y)\ell\le 
(z- y)|C_i| \ell.
$$
Thus, the cost of clustering $\C$ is at most
$(z - y) \ell |\Po|$.
Finally, we turn to the soundness analysis.

\paragraph*{Soundness.}
Consider the optimal $k$-means clustering $\mathcal{C}:=\{C_1,\ldots ,C_k\}$  of the instance (i.e., $C_1\dot\cup\cdots\dot\cup C_k=\Po$). We aim at showing that the $k$-means cost of $\C$ is at
least $((z-y) + 2(1 - \alpha) - o(1))\ell |\Po|$. Given a cluster $C_i$,
let $E_i := \{ T \in E : \tilde A_E(T) \in C_i \}$ be the collection of $z$-sets of $E$ corresponding to $C_i$. 
For each
$S \in {[n] \choose y}$, we define the \emph{degree of $S$ in $C_i$}
to be $$d_{i, S} := 
\left|\{T \mid S \subset T \text{ and } \tA_E(T) \in C_i\}\right|.$$
Let $t_0 = (2z- 2y - \delta) \ell$, $t_1 = (2z - 2y) \ell$, and $t_2 = (2z - 2y + 2 - \delta) \ell$. 
For each cluster $C_i$, let
\begin{align*}
F_i &= \bigg| \{ (p, q) \in C_i^2 : \| p - q \|_2^2 \geq t_2 \} \bigg| \\
M_i &= \bigg| \{ (p, q) \in C_i^2 : \| p - q \|_2^2 \in [t_0, t_1]  \} \bigg| \\
N_i &= \bigg| \{ (p, q) \in C_i^2 : \| p - q \|_2^2 < t_0  \} \bigg|.
\end{align*}
By~\eqref{eqsub} and~\eqref{eqnosub}, $F_i$, $M_i$, and $N_i$ are the number of (ordered) pairs within $C_i$ whose corresponding $z$-sets in the \jc instance intersect in $<y$, $=y$, and $>y$ elements respectively. 
Let $\Delta_i=\underset{S\in\binom{[n]}{y} }{\max}\ d_{i,S}.$
We write the total cost of the clustering as follows.
\begin{align}
& \sum_{i=1}^k W_i  \nonumber
\\\geq \quad & \sum_{i=1}^k \frac{1}{2|C_i|} \bigg( F_i t_2 + M_i t_0 \bigg)  \nonumber
\\=\quad &\sum_{i=1}^k \frac{1}{2|C_i|} \bigg( (|C_i|^2 - M_i) t_2 + (M_i) t_0 - N_i t_2 \bigg)
\label{eq:kmeancost}
\end{align}

We first upper bound $\sum_i (N_i t_2) / (2|C_i|)$. 
For each $z$-set $T$, there are at most $(z-y) \cdot \binom{z}{z - y - 1} \cdot n^{z-y-1}$ sets that intersect with $T$ in at least $y + 1$ elements. 
Therefore, $N_i \leq |C_i| \cdot \max(|C_i|, O(n^{z-y-1}) )$ and 
\[ \sum_{i=1}^k \frac{N_i}{2|C_i|}
 \leq \sum_{i=1}^k \max(|C_i|, O(n^{z-y-1})) 
 \leq  O(k \cdot n^{z-y-1})).
\]
By the definition of \jcd, $|E| = \omega(k \cdot n^{z-y-1})$, 
so $\sum_i N_i t_2 / (2|C_i|)$ is at most $o(\ell |\Po|)$. 

For $M_i$, we prove the following claim that bounds $M_i / |C_i|$ in terms of $\Delta_i$ and $|C_i|$. 
\begin{claim}
For every $i \in [k]$, either $|C_i| = o(|\Po| / k)$ or $M_i / |C_i| \leq (1 + o(1)) \Delta_i + o(|C_i|)$. 
\label{claim:sumdegsquared}
\end{claim}
This proves the theorem because we can lower bound~\eqref{eq:kmeancost} as \allowdisplaybreaks
\begin{align*}
&\sum_{i=1}^k \frac{1}{2|C_i|} \bigg( (|C_i|^2 - M_i) t_2 + (M_i) t_0 - N_i t_2 \bigg)
\\\geq\quad &\sum_{i=1}^k \frac{1}{2|C_i|} \bigg( (|C_i|^2 - M_i) t_2 + (M_i) t_0 \bigg) - o(\ell |\Po|) 
\\ = \quad &\frac{1}{2} \sum_{i=1}^k  \bigg( t_2 |C_i| - (t_2 - t_0 ) (M_i / |C_i|)  \bigg) - o(\ell |\Po|) 
\\ = \quad &\frac{1}{2} \bigg( t_2 |\Po| - (t_2 - t_0 ) \sum_{i=1}^k   (M_i / |C_i|)  \bigg) - o(\ell |\Po|) 
\\ \geq \quad &\frac{1}{2} \bigg( t_2 |\Po| - (t_2 - t_0 ) \sum_{i=1}^k \Delta_i  \bigg) - o(\ell |\Po|), 
\end{align*}
where the last inequality used the fact that either either $M_i / |C_i| \leq (1+o(1))\Delta + o(|C_i|)$ or 
$|C_i| = o(|E| / k)$ where we can use trivial bound $M_i / |C_i| \leq |C_i| = o(|\Po|/k)$ so that all the terms multiplied by $o(1)$ can be absorbed by the last term $o(\ell |\Po|)$.
By using the soundness condition $\sum_{i=1}^k \Delta_i \leq \alpha |\Po|$ and plugging in the values of $t_0$ and $t_2$, the total cost is at least by 
\[
|\Po| \frac{t_0}{2} + |\Po| \frac{(t_2 - t_0) }{2} \bigg( 1 - \alpha  \bigg) - o(\ell |\Po|)
=
\ell |\Po| \bigg( (z - y) + ( 1 - \alpha  ) - \delta/2  - o(1) \bigg).
\]
Therefore, it remains to proof Claim~\ref{claim:sumdegsquared}. 
\begin{proof}[Proof of Claim~\ref{claim:sumdegsquared}]
Note that $M_i \leq \sum_{S \in \binom{[n]}{y}} d_{i, S}^2$. 
  First, consider the set $\S = \{S \in {n\choose y} \mid d_{i,S} < \gamma |C_i|\}$ for some small $\gamma > 0$ that will be determined later, 
  and let $E^0_i = \{T \in E_i \mid \exists S \in \S, S \subseteq T\}$.
  We have that
  \begin{equation}
    \label{eq:lowdegrees}
    \frac{1}{|C_i|} \sum_{S \in \S} d^2_{i,S} \le \frac{\gamma |C_i|}{|C_i|} \sum_{S \in \S} d_{i,S} \le
    \binom{z}{y} \gamma |E^0_i|,
  \end{equation}
  because each $z$-set $T$ can contribute to
  the degree $d_{i,S}$ of at most $\binom{z}{y}$ different sets.

  We then bound $\sum_{S \notin \S} d^2_{i,S}$.
  Let $\S' = {n \choose y} \setminus \S$
  and $E^1_i = \{T \in E_i \mid \exists S_1,S_2 \in \S', S_1 \neq S_2, S_1 \subseteq T, S_2 \subseteq T\}$.
  We claim that
  \begin{equation}
    \label{eq:sumofdegrees}
    \sum_{S \in \S'} d_{i,S} \le \binom{z}{y}|E_i^1|+|E_i|.
  \end{equation}
  Indeed, consider an element $T \in E_i$, we have that either there is at most one element $S$ of $\S'$ such that
  $S \subset T$, in which case $T$ will be counted at most once in $\sum_{S \in \S'} d_{i,S}$, or
  $T$ contains more than one element of $\S'$. 
  The elements of $E^1_i$ are counted at most $\binom{z}{y}$ times in the sum and from there follows
  ~\eqref{eq:sumofdegrees}.

  We claim that if $|E^1_i| > \gamma |E_i|$ or $|E_i|$ is small. 
  From~\eqref{eq:sumofdegrees} we deduce that
  there exists $S \in \S'$ such that $d_{i,S} \le (\binom{z}{y} |E^1_i| + |E_i|)/|\S'|$. Moreover, since $S \in \S'$ we also
  have that $d_{i,S} \ge  \gamma |E_i| \ge  \gamma |E^1_i|$ and so, combining the two bounds yields
  $ \gamma |E^1_i| \le (\binom{z}{y}|E^1_i| + |E_i|)/|\S'|$.
  Rearranging gives
  \[
  |\S'| \le {\binom{z}{y}}/ \gamma  + \frac{|E_i|}{ \gamma |E^1_i|}.
  \]
  If $|E^1_i| >  \gamma |E_i|$ we get $|\S'| \le \binom{z}{y}/\gamma + 1/ \gamma ^2 $. Hence,
  $|E^1_i| $ is at most $(\binom{y}{z}/\gamma + 1 / \gamma ^2)^2 \cdot n^{z - y -1}$ since for each pair $S,S' \in \S'$, $S \neq S'$ there is at most $n^{z - y - 1}$ 
  $T \in E_i$ such that $S \subset T, S' \subset T$.
  Thus, $|E_i| \le (\binom{y}{z}/\gamma + 1/\gamma ^2)^2 / \gamma \cdot n^{z-y-1}$, proving the claim. 

	For $i$ with $|E_i^1| \leq \gamma |E_i|$, 
  using~\eqref{eq:sumofdegrees} together with the above bound on $|E^1_i|$, we obtain
  \begin{equation}
    \label{eq:largedegrees}
    \sum_{S \notin \S} d^2_{i,S} \le \Delta_i \sum_{S \notin \S} d_{i,S} \le \Delta_i (|E_i| + \binom{z}{y}|E^1_i|) \le
    \Delta_i (1+ \gamma \binom{z}{y} )|E_i|.
  \end{equation}  
Combining~\eqref{eq:lowdegrees} and~\eqref{eq:largedegrees} we deduce that
  $\frac{1}{|E_i|}\sum_{S} d^2_{i,S} \le  (1+\gamma \binom{z}{y})\Delta_i + \gamma \binom{z}{y} |E_i|$.

Finally, since $|E| = \omega(kn^{z-y-1})$, we can choose $\gamma = o(1)$ such that 
\[
\gamma \binom{z}{y} = o(1) \mbox{ and } \bigg( \bigg(\binom{y}{z} / \gamma + 1 / \gamma^2 \bigg)^2/ \gamma \bigg) \cdot k n^{z-y-1} = o(|E|).
\]
Therefore, if $|E_i^1| \leq \gamma |E_i|$, then $M_i / |E_i| \leq (1 + o(1))\Delta_i + o(|E_i|)$,
and if $|E_i^1| > \gamma |E_i|$, $|E_i| = o(|E| / k)$. 
This finishes the proof of the claim and the theorem.
\end{proof}
\end{proof}

\begin{remark}\label{rem:APXcont}
Similar to Remark~\ref{rem:APXdisc}, we note here that given $(1-\delta,z,1)$-Johnson Coverage problem is \NP-Hard   for some constants $\delta>0$ and $z\in\NN$, plugging those parameters in Theorem~\ref{thm:kmeanl2}, we immediately have that continuous \kmean is \APX-Hard in   $\ell_2$-metric. \end{remark}

Before, we address continuous case of \kmed in $\ell_1$-metric, we prove hardness of approximating the continuous case of both \kmean and 
\kmed in the Hamming metric. The inapproximability of \kmed in $\ell_1$-metric then follows by a simple observation. 

\begin{theorem}[\kmed without candidate centers in $O(\log n)$ dimensional $\ell_0$-metric space]\label{thm:kmedHam}
Assume $(\alpha, z, y)$-\jcd is \NP-Hard. 
For every constant $\varepsilon>0$,  given a point-set $\Po\subset \{0,1\}^{d}$ of size $n$ (and $d=O(\log n)$) and a parameter $k$ as input, it is \NP-Hard to distinguish between the following two cases:
\begin{itemize}
\item \textbf{\emph{Completeness}}:  
There exists $\C':=\{c_1,\ldots ,c_k\}\subseteq \mathbb \{0,1\}^d$ and $\sigma:\Po\to\C'$ such that $$\sum_{a\in\Po}\|a-\sigma(a)\|_0\le \rho n\log n,$$
\item \textbf{\emph{Soundness}}: For every $\C':=\{c_1,\ldots ,c_k\}\subseteq \{0,1\}^d$ and every $\sigma:\Po\to\C'$ we have: $$\sum_{a\in\Po}\|a-\sigma(a)\|_0\ge \left(1 + \frac{(1 - \alpha)}{(z-y)} - \eps \right)\cdot \rho n\log n,$$
\end{itemize}
for some constant $\rho>0$.
\end{theorem}

\begin{theorem}[\kmean without candidate centers in $O(\log n)$ dimensional $\ell_0$-metric space]\label{thm:kmeanHam}
Assume $(\alpha, z, y)$-\jcd is \NP-Hard. 
For every constant $\varepsilon>0$,  given a point-set $\Po\subset \{0,1\}^{d}$ of size $n$ (and $d=O(\log n)$) and a parameter $k$ as input, it is \NP-Hard to distinguish between the following two cases:
\begin{itemize}
\item \textbf{\emph{Completeness}}:  
There exists $\C':=\{c_1,\ldots ,c_k\}\subseteq \mathbb \{0,1\}^d$ and $\sigma:\Po\to\C'$ such that $$\sum_{a\in\Po}\|a-\sigma(a)\|_0^2\le \rho n(\log n)^2,$$
\item \textbf{\emph{Soundness}}: For every $\C':=\{c_1,\ldots ,c_k\}\subseteq \{0,1\}^d$ and every $\sigma:\Po\to\C'$ we have: $$\sum_{a\in\Po}\|a-\sigma(a)\|_0^2\ge \left(
1 + \frac{(1 - \alpha) (2(z - y) + 1)}{(z-y)^2}  -\varepsilon\right)\cdot \rho n(\log n)^2,$$
\end{itemize}
for some constant $\rho>0$.
\end{theorem}

\begin{proof}[Proof of Theorems \ref{thm:kmedHam} and \ref{thm:kmeanHam}]
The construction of the point-set $\Po$ in the theorem statement is the same as in the proof of Theorem~\ref{thm:kmeanl2}. 
 The completeness case is the same as in Theorem~\ref{thm:kmeanellp}. Therefore, we have that  the \kmean cost of the overall instance is at most $ \ell^{2}\cdot |\Po|$, while the \kmed cost is at most
  $\ell\cdot|\Po|$. 
The soundness analysis to prove the theorem follows by a case analysis, which is elaborated in \cite{CK19} and we skip it here for the sake of brevity.  
\end{proof}

We consider below \kmed in $\ell_1$-metric. 

\begin{theorem}[\kmed without candidate centers in $O(\log n)$ dimensional $\ell_1$-metric space]\label{thm:kmedl1}
Assume $(\alpha, z, y)$-\jcd is \NP-Hard. 
For every constant $\varepsilon>0$,  given a point-set $\Po\subset \{0,1\}^{d}$ of size $n$ (and $d=O(\log n)$) and a parameter $k$ as input, it is \NP-Hard to distinguish between the following two cases:
\begin{itemize}
\item \textbf{\emph{Completeness}}:  
There exists $\C':=\{c_1,\ldots ,c_k\}\subseteq \mathbb \R^d$ and $\sigma:\Po\to\C'$ such that $$\sum_{a\in\Po}\|a-\sigma(a)\|_1\le \rho n\log n,$$
\item \textbf{\emph{Soundness}}: For every $\C':=\{c_1,\ldots ,c_k\}\subseteq \R^d$ and every $\sigma:\Po\to\C'$ we have: $$\sum_{a\in\Po}\|a-\sigma(a)\|_1\ge \left(1 + \frac{(1 - \alpha) }{(z-y)} - \eps \right)\cdot \rho n\log n,$$
\end{itemize}
for some constant $\rho>0$.
\end{theorem}
\begin{proof}[Proof Sketch]
  The proof follows from a simple observation that for instances where all
  the coordinates of the points to be clustered are in $\{0,1\}$, we have
that for any subset (i.e. cluster) of the points of the instance,
  an optimal center of the set is such that its 
  $i$th coordinate is the median
  of a set of values in $\{0,1\}$, i.e., we may assume without loss of generality that the $i$th coordinate of the center is also in $\{0,1\}$. Therefore, simply  mimicking
  the proof of Theorem~\ref{thm:kmedHam} yields the desired theorem statement.
See \cite{CK19} for a formal argument.   
  \end{proof}


Finally, we finish the section by proving the remaining hardness results for continuous \kmed and \kmean:
\kmed in $\ell_2$-metric and 
\kmean in $\ell_1$-metric.
Especially, \kmed in $\ell_2$-metric for $k = 1$ is known as the {\em geometric median} and is  actively studied for bounded $d$ or $k$~\cite{arora1998approximation, badoiu2002approximate, kumar2010linear, feldman2011unified}, but to the best of our knowledge, no (hardness of) approximation algorithms have been studied for general $d$ and $k$. 
First, we prove the following lemma that if points are pairwise far apart, there is no good center that is close to all of them. 

\begin{lemma}
For any $\eps \in (0, 1)$, 
if $p_1, \dots, p_n \in \R^d$ with $n > 4/\eps + 1$ satisfy $\| p_i - p_j \|_2 \geq \sqrt{2}$, then, for any $c \in \R^d$, 
$\max_{i \in [n]} \| c - p_i \|_2 \geq 1 - \eps$. 
\label{lem:kmed-l2}
\end{lemma}
\begin{proof}
Assume towards contradiction that there exists $c$ such that $\max_{i \in [n]} \| c - p_i \|_2 < 1 - \eps$. 
Without loss of generality let $c = \vec{0}$ and consider the matrix $A \in \R^{n \times n}$ such that 
$A_{i,j} = \langle p_i, p_j \rangle = (\| p_i \|_2^2 + \| p_j \|_2^2 - \| p_i - p_j \|_2^2) / 2 < (2(1 - \eps)^2 - 2) / 2 \leq \eps$. 
Then $A$ is positive semidefinite, so if $\vec{1}$ denotes the all-ones vector, 
\[
(\vec{1}^T)A\vec{1} \leq \sum_{i} A_{i,i} + \sum_{i < j} A_{i,j} \leq n - \binom{n}{2} \eps \geq 0,
\]
which leads to contradiction if $n > 4/\eps + 1$. 
\end{proof}

\begin{theorem}[\kmed without candidate centers in $O(\log n)$-dimensional $\ell_2$-metric space]\label{thm:kmedl2}
Assume $(\alpha, z, y)$-\jcd is \NP-Hard. 
For every constant $\varepsilon>0$,  given a point-set $\Po\subset \R^{d}$ of size $n$ (and $d=O(\log n)$) and a parameter $k$ as input, it is \NP-Hard to distinguish between the following two cases: 
\begin{itemize}
\item \textbf{\emph{Completeness}}:  
There exists $\C':=\{c_1,\ldots ,c_k\}\subseteq \mathbb \R^d$ and $\sigma:\Po\to\C'$ such that $$\sum_{a\in\Po}\|a-\sigma(a)\|_2\le n$$
\item \textbf{\emph{Soundness}}: For every $\C':=\{c_1,\ldots ,c_k\}\subseteq \R^d$ and every $\sigma:\Po\to\C'$ we have: $$\sum_{a\in\Po}\|a-\sigma(a)\|_2\ge \left(\alpha+ \sqrt{\frac{z-y+0.5}{z-y}}(1-\alpha) - \eps \right)n$$
\end{itemize}
\end{theorem}
\begin{proof} 
Given an instance $(U, E, k)$ of $(\alpha, z, y)$-\jcd, for any $S \subseteq U$, let $\tau(S)$ be the indicator vector of $S$. (The dimension can be reduced to $O(\log n)$ by the standard dimension reduction technique.) 

In the completeness case, every point is connected to a center at distance $\sqrt{z - y}$.
For the soundness case, fix one cluster $C$ and its center $c \in \R^U$.
Fix $\eps > 0$, and remove all points from $C$ whose $\ell_2$ distance to $c$ is greater than $(\sqrt{z - y + 0.5} - \eps)$.
Assume $C = \omega(n^{z-y-1})$. 
We will show that there exists $S' \in \binom{U}{y}$ that covers all but $O(n^{z-y-1})$ sets in $C$.

We claim that no two $y$-sets can both cover $\omega(n^{z-y-1})$ number of $z$-sets in $C$. 
Assume towards contradiction that there exist $S', T' \in \binom{U}{y}$ be such that each of them covers at least $t = \omega(n^{z-y-1})$ sets in $C$. 
Let $s = |S' \cap T'|$ and $I = S' \cap T'$. 
Consider the center $c$ and its squared distances to $\tau(S')$ and $\tau(T')$ in coordinates $Q = (S' \cup T') \setminus I$. 
Since $\tau(S')$ and $\tau(T')$ are different in these coordinates, $(c_i - \tau(S')_i)^2 + (c_i - \tau(T')_i)^2 \geq 0.5$ for each $i \in Q$.
Without loss of generality, assume that $(c_i - \tau(S')_i)^2 \geq (y-s) / 2 \geq 0.5$. 

Now let us consider the points restricted to the coordinates $O = U \setminus (S \cup T)$. 
First, we claim that we can choose $S_1, \dots, S_{\ell}$ with $\ell = \omega(1)$ such that $S_i \in C$, $S' \subseteq S_i$, $S_i \setminus S \subseteq O$, and $S_i$'s are pairwise disjoint. First the number of $S \in C$ that contains $S'$ and intersects $T'$ with in at least one element outside $S'$ is at most $|T'| \cdot n^{z-y-1}$, which is only an $o(1)$ fraction of the $z$-sets in $C$ covered by $S'$. Remove them from $C$. Now, consider a greedy iteration for $i = 1, \dots, \ell$ where we pick an arbitrary set $S_i \in C$ that contains $S'$, and remove all sets from $C$ that insect $S_i$ outside $S'$. Each time, the number of removed sets is only $(z - y)n^{z-y-1}$, and since $C = \omega(n^{z-y-1})$, this process can continue $\ell = o(1)$ iterations. 

Then $\tau(S_i)$ and $\tau(S_j)$ on coordinates in $O$  are at distance at least $\sqrt{2(z-y)}$ from each other. 
Since the contribution of $\| c - S_i \|_2^2$ from $Q$ is already at least $0.5$, 
Lemma~\ref{lem:kmed-l2} shows there is no center can cover every point in $C$ at distance at most $\sqrt{z - y + 0.5} - o(1)$, leading to the desired contradiction. 


Now we prove that there exist a bounded number of $y$-sets that collectively cover every $z$-sets. 
Consider the following process starting with $C' = C$ and $i = 1$. 
\begin{enumerate}
\item Let $S_i$ be an arbitrary set from $C'$. 
\item Delete all $S \in C'$ such that $|S \cap S_i| \geq y$. 
\item If $C'$ is nonempty, increase $i$ by $1$ and repeat from 1. 
\end{enumerate}
Let $t$ be the final value of $i$, and consider $S_1, \dots, S_t$. They are at distance at least $\sqrt{2(z - y + 1)}$ from each other,
so again by Lemma~\ref{lem:kmed-l2} and using that $\sqrt{2(z-y+1)} / \sqrt{2} = \sqrt{z-y + 1} >  \sqrt{z-y + 0.5}$, $t$ is most some absolute constant.
Then $y$-sets in $\{ S' \subseteq U : |S'| = y\mbox{ and } S' \subseteq S_i\mbox{ for some } i\in[t] \}$ cover all sets in $C$. 
A similar argument for $C^* = \{ S : \| \tau(S) - c \|_2 < \sqrt{z - y} - \eps \}$ and $y + 1$ implies that $|C^*| \leq O(n^{z-y-1})$. 

So far, we showed that (1) there are a constant number of $y$-sets that collectively cover every $S \in C$
(2) except one set, every $y$-set can cover at most $O(n^{z-y-1})$ sets in $C$,
and (3) $|C^*| \leq O(n^{z-y-1})$. 
Therefore, whenever $|C| = \omega(n^{z-y-1})$, there exists $S' \in \binom{U}{y}$ that covers all but a subconstant fraction of sets in $C$. 
Since $|E| = \omega(k n^{y-z-1})$, we can use an argument similar to Theorem~\ref{thm:kmeanl2} to show that in the soundness case, the fraction of points that are covered by a center at distance at most $\sqrt{z-y+0.5} - \eps$ is at most $\alpha + o(1)$, and at most an $o(1)$ fraction of points are covered at distance at most $\sqrt{z - y} - \eps$. 
This proves the theorem. 
\end{proof}

\begin{theorem}[\kmean without candidate centers in $\poly(n)$-dimensional $\ell_1$-metric space]\label{thm:kmeanl1}
Assume $(\alpha, z, y)$-\jcd is \NP-Hard. 
For every constant $\varepsilon>0$,  given a point-set $\Po\subset \{0,1\}^{d}$ of size $n$ (and $d=O(\log n)$) and a parameter $k$ as input, it is \NP-Hard to distinguish between the following two cases:
\begin{itemize}
\item \textbf{\emph{Completeness}}:  
There exists $\C':=\{c_1,\ldots ,c_k\}\subseteq \mathbb \R^d$ and $\sigma:\Po\to\C'$ such that $$\sum_{a\in\Po}\|a-\sigma(a)\|_1^2\le n$$
\item \textbf{\emph{Soundness}}: For every $\C':=\{c_1,\ldots ,c_k\}\subseteq \R^d$ and every $\sigma:\Po\to\C'$ we have: $$\sum_{a\in\Po}\|a-\sigma(a)\|_1^2\ge \left( \alpha +(1 - \alpha)\frac{(z-y+1)^2}{(z-y)^2}-\varepsilon\right)n.$$
\end{itemize}
\end{theorem}
\begin{proof}
Given an instance $(U, E, k)$ of $(\alpha, z, y)$-\jcd, for any $S \subseteq U$, let $\tau(S)$ be the indicator vector of $S$. We analyze the completeness and soundness of this simple reduction. 
Since every point is in the boolean hypercube, the embedding of~\ref{thm:kmeanellp} ensures that one can reduce the dimension to $O(\log n)$ with losing an arbitrarily small constant in the inapproximability factor. 

In the completeness case, every point is connected to a center at distance $z-y$. 
For the soundness case, fix one cluster $C$ and its center $c \in \R^U$.
Fix $\eps > 0$, and remove all points from $C$ whose $\ell_1$ distance to $c$ is at least $(z - y + 1 - \eps)$.

First, note that for every $S, T \in C$, $|S \cap T| \geq y$. Otherwise, $\| \tau(S) - \tau(T) \|_1 \geq 2(z-y+1)$, so there is no center that are at distance strictly less than $z-y+1$ from both $\tau(S)$ and $\tau(T)$. 

Fix any $S \in C$. Then every $T \in C$ is covered by a $y$-set in $\mathcal{F} = \{ S' \subseteq U : |S'| = y\mbox{ and } S' \subseteq S \}$. 
Furthermore, argument similar to the proof of Theorem~\ref{thm:kmedl2} show except possibly one, every $y$-set can cover at most $O(n^{z-y-1})$ sets in $C$,
and the number of points that are covered at distance strictly less than $z-y$ is also at most $O(n^{z-y-1})$. 

Therefore, whenever $|C| = \omega(n^{z-y-1})$, there exists $S' \in \binom{U}{y}$ that covers all but a subconstant fraction of sets in $C$. 
Since $|E| = \omega(k n^{y-z-1})$, we can use an argument similar to Theorem~\ref{thm:kmeanl2} to show that in the soundness case, the fraction of points that are covered by a center at distance at most $z-y+1 - \eps$ is at most $\alpha + o(1)$, and at most an $o(1)$ fraction of points are covered at distance at most $z - y - \eps$. This proves the theorem. 
\end{proof}

\subsection{Integrality Gaps}
\label{section:gaps}

In this subsection we present improved integrality gaps. 
Given an instance of discrete \kmed or \kmean with a set of candidate centers $\C$, a set of points $\Po$, 
and distances $\{ d_{p, c } \}_{p \in \Po, c \in \C}$, 
the basic LP relaxation for \kmed or \kmean is the following. (For \kmean, $d_{p,c}$ becomes a squared distance.)

\begin{align*}
\mbox{Minimize} & \quad \sum_{p \in \Po} \sum_{c \in \C} x_{p, c} \cdot d_{p, c} \\
\mbox{Subject to} & \quad \sum_{c \in C} x_{p, c} = 1 & \forall p \in \Po \\
& \quad x_{p, c} \leq y_{c}  & \forall p \in \Po, c \in \C \\
& \quad \sum_{c \in C} y_c \leq k  \\
& \quad x, y \geq 0. 
\end{align*}

The basic SDP relaxation, which replaces $x_{p, c}$ and $y_{c}$ by $\| v_{p,c}\|_2^2$ and $\| u_c \|_2^2$ for some vectors $\{ v_{p, c} \}$ and $\{ u_c \}$ with additional constraints, is the following. 

\begin{align}
\mbox{Minimize} & \quad \sum_{p \in \Po} \sum_{c \in \C} {\| v_{p, c} \|_2^2 \cdot d_{p, c}}  \nonumber \\ 
\mbox{Subject to} & \quad \langle v_0, v_0 \rangle = 1 \nonumber \\
& \quad \langle v_{p, c}, v_0 \rangle = \| v_{p, c} \|_2^2 & \forall p \in \Po, c \in \C \label{eq:vlength} \\
& \quad \langle u_{c}, v_0 \rangle = \| u_{c} \|_2^2 & \forall c \in \C \label{eq:ulength} \\
& \quad \langle v_{p, c}, u_{c} \rangle = \|v_{p,c}\|_2^2 & \forall p \in \Po, c \in \C \label{eq:uv} \\
& \quad \| \sum_{c \in C} v_{p, c} - v_0 \|_2^2 = 0 & \forall p \in \Po \label{eq:localconst} \\
& \quad \sum_{c \in C} \|  y_c \|_2^2 \leq k \nonumber
\end{align}

We give stronger integrality gap instances than the standard notion. 
The LP or SDP relaxation has a {\em robust} gap of $\alpha > 1$ 
if there exists a family of instances for infinitely many values of $k$
where for any $\eps > 0$, there exists $\delta > 0$ and $k_0 \in \NN$ such that 
for all instances in the family with $k \geq k_0$, 
the gap between the optimal fractional solution and the optimal integral solution is at least $\alpha - \eps$ 
even when the optimal integral solution is allowed to open $k + \delta k$ facilities. 

Our instances additionally satisfy {\em well-separated} the property that every pair of points in $\binom{\Po \cup \C}{2}$ are at distance at least $1$ far part,
where each instance is normalized so that $1$ is the average connection cost of each point in the SDP solution. 

The basic LP relaxation for metric \kmed has a gap of $2$ under the standard notion~\cite{jain2002new}, but this instance does not give a robust gap.
The best approximation algorithms for \kmed and \kmean in both general and Euclidean metrics bound the robust gap of the LP relaxation~\cite{BPRST15, ANSW16}, and to the best of our knowledge, for well-separated instances, no robust gap of the LP relaxation better than computational hardness results~\cite{GuK99, CK19} were previously known. 

\begin{theorem}
Fix any $\eps > 0$. 
For discrete \kmed in $\ell_1$ and discrete \kmean in $\ell_2$, there is a family of well-separated instances where
the SDP relaxation has a gap of at least $149/125 - \eps \approx 1.192 - \eps$, even when the integral solution opens $\Omega(k)$ more centers.
\label{theorem:gaps}
\end{theorem}
\begin{proof}
Consider the complete graph $K_n = ([n], E)$ with $E = \binom{[n]}{2}$.  
Let $\tau : 2^{E} \to \R^{n}$ where for any $p \subseteq E$, $\tau(p)$ is defined to be the characteristic vector of $\cup_{e \in p} e$. 
The set points $\Po$ is defined to be $\Po := \{ \tau( p ) : p \mbox{ forms a 4-clique in }K_n \}$.
The set $C$ of candidate centers, for each edge $e \in E$, has 
a center $\tau(e)$. 
To simplify notation, we use $p$ (resp. $e$) to also denote $\tau(p)$ (resp. $\tau(e)$) when $p$ (resp. $e$) is used as part of the clustering instance. 

We say that an edge $e \in E$ {\em covers} a $4$-clique $p$ if $e \in p$; 
in this case a center $e \in \C$ also {\em covers} a point $p \in \Po$.
Note that $\| e - p \|_0 = \| e - p \|_1 = 2$ if $e$ covers $p$ and at least $4$ otherwise. 
Since each $4$-clique has 6 edges and each edge yields $m$ centers, note that each point $p \in \Po$ is covered by exactly $6$ centers. 
Therefore, there is an LP solution that picks  $\nicefrac{1}{6}$ of each $e \in C$ and fractionally covers every $p \in \Po$.
We prove that even for the SDP, there is such a solution that picks $\nicefrac{1}{5}$ of each $e \in C$. 
\begin{claim}
There exists a feasible SDP solution where every $p \in \Po$ is fractionally connected to only centers at distance at most $2$, and $\| u_{e} \|_2^2 = \nicefrac{1}{5}$ for every $e$. 
\end{claim}
\begin{proof}
Let $t = 5$. 
Let $v_0$ and $\{ w_{e}, w'_{e} \}_{e \in E}$ be pairwise orthogonal unit vectors. We explicitly construct the vectors for the SDP relaxation.
\begin{itemize}
\item For each $e \in E$, $u_{e} = \frac{v_0}{t} + \big(\frac{(t-1)\sqrt{t+1}}{t^2} \big) w_{e} + \big( \frac{\sqrt{t-1}}{t^2}\big) w'_{e}$. 
\item For each $e \in E$, and a $4$-clique $p$, if $e \in p$, 
\[
v_{p, e} = \frac{1}{t+1} \cdot v_0 + \frac{t}{(t+1)^{3/2}} \cdot w_{e} - \sum_{f \in p \setminus \{ e \}} \frac{1}{(t+1)^{3/2}} \cdot w_{f}.\]
Otherwise $v_{p, e} = 0$. 
\end{itemize}
Note that every point $p \in \Po$ is only connected to $e$ that covers $p$. 
We check each constraint of the SDP. The first constraint $\langle v_0, v_0 \rangle = 1$ is true by definition. 
Since 
\[
\| u_{e} \|_2^2 = \frac{1}{t^2} + \frac{(t-1)^2(t+1)}{t^4} + \frac{t-1}{t^4} = 
\frac{1}{t^2} + \frac{t-1}{t^2} = \frac{1}{t} = \langle v_0, u_{e} \rangle,
\]
it satisfies~\eqref{eq:ulength}. 
Since for $e \in p$, 
\[
\| v_{p, e} \|_2^2 = \frac{1}{(t+1)^2} + \frac{t^2}{(t+1)^3} + \frac{t}{(t+1)^3} = 
\frac{1}{t + 1} = \langle v_0, v_{p, e} \rangle,
\]
it satisfies~\eqref{eq:vlength}. Furthermore, for each $p \in \Po$, 
$\sum_{e \in p} v_{p, e} = v_0$, since the coefficient of $w_{e}$ in the sum for any $e \in p$  is $\frac{t}{(t+1)^{3/2}}$ (from $v_{p, e}$) minus $t$ times $\frac{1}{(t+1)^{3/2}}$ (from every other $v_{p, f}$), which is $0$. It satisfies~\eqref{eq:localconst}. 
Finally,~\eqref{eq:uv} can be checked as for every $4$-clique $p$, $e \in p$,
\[
\langle u_{e}, v_{p, e} \rangle = \frac{1}{t}\cdot \frac{1}{t+1} + \frac{(t-1)\sqrt{t+1}}{t^2} \cdot \frac{t}{(t+1)^{3/2}} 
= \frac{1}{t(t+1)} + \frac{t-1}{t(t+1)} = \frac{1}{t+1} = \| v_{p, e} \|_2^2. 
\]
Since the solution satisfies every SDP constraint, each $p \in \Po$ is only connected to a center covering $p$, and $\| u_{e} \|_2^2 = 1/t$ for every $e$, the claim is proved. 
\end{proof}

Therefore the optimal relaxation value is at most $2\cdot \binom{n}{4}$ when $k = \binom{n}{2} / 5$.
Now we consider if we pick at most $k$ edges $E$ integrally, how many $4$-cliques can be covered by them. 
Equivalently, we ask if we pick at least $\binom{n}{2} - k$ edges, how many $4$-cliques are completely contained by them. 
The clique density theorem of Reiher~\cite{reiher2016clique} answers this question, 
proving that if we pick at least an $(1 - \nicefrac{1}{t})$ fraction of edges for some integer $t \geq 4$, the number of $4$-cliques completely contained in them is at least 
that of complete $t$-partite graph with each partition having the same size. 
Note that in the complete $t$-partite graph, the probability that a random $4$-tuple becomes a clique is roughly $1 \cdot \frac{t - 1}{t} \cdot \frac{t - 2}{t} \cdot \frac{t - 3}{t}$. 

\begin{theorem}[\cite{reiher2016clique}]
Let $t \geq 4$ be an integer. In every graph on $n$ vertices with at least $\frac{(t - 1)}{t} \cdot \binom{n}{2}$ edges,
the number of 4-cliques is at least 
\[
\binom{n}{4} \cdot \frac{t(t-1)(t-2)(t-3)}{t^4}.
\]
\end{theorem}

Applying the above theorem with $t = 5$ shows that if we pick $k$ edges, the fraction of 4-cliques covered by them is at most 
$1 - \nicefrac{(5\cdot4\cdot3\cdot2)}{5^4} = \nicefrac{24}{125}$. 
Therefore, while every point is connected at distance $2$ the SDP solution, in the integral solution at least a $\nicefrac{24}{125}$ fraction of the points are connected at distance at least $4$. The gap is at least $(2 + 2(24/125) / 2) = 149/125 \approx 1.192$. 
\end{proof}

\section{NP-Hardness of Approximatig 3-Hypergraph Vertex Coverage problem}\label{sec:nphard}

In this section, we prove the following theorem showing that for any $\eps > 0$, $(7/8  + \eps, 3, 1)$-Johnson Coverage Problem is NP-hard for randomized reductions (even in the dense case). 
By the results in Sections~\ref{sec:condDisc}~and~\ref{sec:condCont}, we obtain the inapproximability results for clustering problems stated in Theorem~\ref{thm:np-informal} and Table~\ref{table}. 

\begin{theorem}
For any $\eps > 0$, given a simple 3-hypergraph $\calh = (V, H)$ with $n = |V|$, it is NP-hard to distinguish between the following two cases:
\label{thm:np-hard}
\begin{itemize}
\item {\bf Completeness:} There exists $S \subseteq V$ with $|S| = n/2$ that intersects every hyperedge. 
\item {\bf Soundness:} Any subset $S \subseteq V$ with $|S| \leq n/2$ intersects at most a $(7/8 + \eps)$ fraction of hyperedges. 
\end{itemize}
Furthermore, with randomized reductions, the above hardness holds when $|H| = \omega(n^2)$. 
\label{thm:hvc}
\end{theorem}

We first prove Theorem~\ref{thm:hvc} without the lower bound on $|H|$; 
Section~\ref{subsec:hvc_construction} presents the construction and proves the completeness of the reduction, 
and Section~\ref{sec:hvc_soundness},~\ref{sec:hvc_proof_1}, and~\ref{sec:hvc_proof_2} analyze its soundness.
Finally, Section~\ref{subsec:dense} shows how to ensure $|H| = \omega(n^2)$.

\subsection{Overview and Construction}
\label{subsec:hvc_construction}
Consider the setting of Theorem~\ref{thm:hvc}.
A random set of $|V|/2$ elements will intersect each hyperedge with probability $7/8$, so the theorem says that it is hard to do even slightly better than the random solution.
It is similar to the notion of {\em approximate resistance} mainly studied for constraint satisfaction problems. 
Indeed, given an instance of Max 3-SAT with variables $\{x_1, \dots, x_n \}$ and clauses $\{ (\ell_{i,1} \cup \ell_{i,2} \cup \ell_{i,3}) \}_{i \in [m]}$ where each $\ell_{i, j}$ denotes a literal $x_k$ or $\overline{x_k}$, consider the simple reduction of creating a vertex for each literal and a hyperedge $(\ell_{i,1}, \ell_{i,2}, \ell_{i,3})$ for each clause. 
This reduction almost proves the above theorem, except that the soundness property only holds for $S$ that satisfies $|S \cap \{ x_i, \overline{x_i} \}| \leq 1$ for each $i \in [n]$. 
However, the resulting hypergraph produced by this reduction combined with \hastad's celebrated hardness for Max 3-SAT~\cite{H01} is always bipartite due to the underlying bipartite structure of the outer verifier, so there is a vertex cover of size at most $|V| / 2$ even in the soundness case. 

We bypass the above problem by plugging in \hastad's inner verifier to the outer verifier constructed in Guruswami et al.~\cite{DGKR05} and Khot~\cite{Khot02}. 
This outer verifier is called a {\em multilayered PCP} and used for proving hardness of the covering version of our problem called $k$-Hypergraph Vertex Cover. 
We give a new analysis for the multilayered PCP and combine the analysis for \hastad's inner verifier to bound the number of uncovered hyperedges for all $S$ with $|S| \leq |V|/2$. 

We now formally present the reduction. We first describe multilayered PCPs that we use. 

\begin{definition}
An $\ell$-layered PCP $\calm$ consists of
\begin{itemize}
\item An $\ell$-partite graph $G = (V, E)$ where $V = \cup_{i=1}^{\ell} V_i$. Let $E_{i,j} = E \cap (V_i \times V_j)$. 
\item Sets of alphabets $\Sigma_1, \dots, \Sigma_{\ell}$. 
\item For each edge $e = (v_i, v_j) \in E_{i,j}$, a surjective projection $\pi_e : \Sigma_j \to \Sigma_i$. 
\end{itemize}
Given an assignment $(\sigma_i : V_i \to \Sigma_i)_{i \in [\ell]}$, an edge $e = (v_i, v_j) \in E_{i,j}$ is {\em satisfied} if $\pi_e(\sigma_j(v_j)) = \sigma_i(v_i)$. 
There are additional properties that $\calm$ can satisfy.
\begin{itemize}
\item $\eta$-smoothness: For any $i < j$, $v_j \in V$, and $x, y \in \Sigma_j$, $\Pr_{(v_i, v_j) \in E_{i,j}} [\pi_{(v_i, v_j)}(x) = \pi_{(v_i, v_j)}(y)] \leq \eta$. 
\item Path-regularity: Call a sequence $p = (v_1, \dots, v_{\ell})$ {\em full path} if $(v_i, v_{i+1}) \in E_{i, i+1}$ for every $1 \leq i < \ell$,
and let $\calp$ be the distribution of full paths obtained by (1) sampling a random vertex $v_1 \in V_1$ and (2) for $i = 2, \dots, \ell$, sampling $v_i$ from the neighbors of $v_{i-1}$ in $E_{i-1, i}$. 
$\calm$ is called {\em path-regular} if for any $i < j$, sampling $p = (v_1, \dots, v_{\ell})$ from $\calp$ and taking $(v_i, v_j)$ is the same as sampling uniformly at random from $E_{i,j}$. 
\end{itemize}
\end{definition}

\begin{theorem} [\cite{DGKR05, Khot02}]
For any $\tau, \eta > 0$ and $\ell \in \NN$, given an $\ell$-layered PCP $\calm$ with $\eta$-smoothness and path-regularity, it is NP-hard to distinguish between the following cases.
\begin{itemize}
\item {\bf Completeness: } There exists an assignment that satisfies every edge $e \in E$.
\item {\bf Soundness: } For any $i < j$, no assignment can satisfy more than an $\tau$ fraction of edges in $E_{i,j}$. 
\end{itemize}
\label{thm:dgkr}
\end{theorem}

Given an $\ell$-layered PCP $\calm$ described as above, we construct the reduction from to Johnson Coverage.
For simplicity of presentation, the produced instance will be vertex-weighted and edge-weighted, so that the problem becomes ``choose a set of vertices of total weight at most $k$ to maximize the total weight of covered edges.''
Vertex weights can be easily removed by duplicating vertices according to weights and creating edges between duplicated vertices with appropriate weights. 
Edge weights will be handled in Section~\ref{subsec:dense}.

\begin{itemize}
\item Let $C_i := \{ \pm 1 \}^{\Sigma_i}$ and $U_i := V_i \times C_i$. 
The resulting hypergraph will be denoted by $\calh = (U, H)$ where $U = \cup_{i=1}^{\ell} (V_i \times C_i)$. The weight of vertex $(v, x) \in V_i \times C_i$ is 
\[
w(v, x) :=\frac{1}{\ell} \cdot \frac{1}{|V_i|} \cdot \frac{1}{|C_i|}.
\]
Note that the sum of all vertex weights is $1$. 

\item Let $\cald_I$ be the distribution where $i \in [\ell]$ is sampled with probability $(\ell-i)^2 / (6\ell(\ell - 1)(2\ell - 1))$, and $\cald$ be the distribution over $(i, j) \in [\ell]^2$ where $i$ is sampled from $\cald_I$ and $j$ is sampled uniformly from $\{ i+1, \dots, \ell \}$. For each $i < j$, we create a set of hyperedges $H_{i,j}$ that have one vertex in $U_i$ and two vertices in $U_j$. 
Fix each $e = (v_i, v_j) \in E_{i,j}$ and a set of three vertices $t \subseteq (\{ v_i \} \times C_i) \cup (\{ v_j \} \times C_j)$. The weight $w(t)$ is $($the probability that $(i, j)$ is sampled from $\cald) \cdot (1/|E_{i,j}|) \cdot ($the probability that $t$ is sampled from the following procedure$)$. The reduction is parameterized by $\delta > 0$ determined later. 
\begin{itemize}
\item For each $a \in \Sigma_i$, sample $x_a \in \{ \pm 1 \}$.
\item For each $b \in \Sigma_j$, 
\begin{itemize}
\item Sample $y_b \in \{ \pm 1 \}$. 
\item If $x_{\pi(b)} = -1$, let $z_b = y_b$ with probability $1-\delta$ and $z_b = -y_b$ otherwise.
\item If $x_{\pi(b)} = 1$, let $z_b = -y_b$. 
\end{itemize}
\item Output $\{ (v_i, x), (v_j, y), (v_j, z) \}$. 
\end{itemize}
Note that the sum of all hyperedge weights is also $1$. 
\end{itemize}

\paragraph{Completeness.}
If $\calm$ admits an assignment $(\sigma_i : V_i \to \Sigma_i)_{i \in [\ell]}$ that satisfies every edge $e \in E$, 
let $S := \{ (v_i, x) : v_i \in V_i, x_{\sigma_i(v_i)} = -1 \}$. 
Fix any $e = (v_i, v_j) \in E_{i,j}$ and consider the above sampling procedure to sample $x \in \{ \pm 1 \}^{\Sigma_i}$ and $y \in \{ \pm 1 \}^{\Sigma_j}$ when $b = \sigma_j(v_j)$. 
Since $\pi_e(\sigma_j(v_j)) = \sigma_i(v_i)$, at least one of $x_{\sigma_i(v_i)}$, $y_{\sigma_j(v_j)}$, $z_{\sigma_j(v_j)}$ must be $-1$ always. This proves that $S$ intersects every hyperedge with nonzero weight.

\subsection{Soundness.}
\label{sec:hvc_soundness}
In the soundness case, we want to prove that any subset of weight at most $1 / 2$ intersects hyperedges of total weight at most $7/8 + o(1)$.
We prove the equivalent statement that for any $S \subseteq V$ of weight {\em greater than} $1 / 2$ contains hyperedges of total weight approximately {\em at least} $1/8 - o(1)$.

Fix a set $S \subseteq V$ with $w(S) \geq 1 / 2$. Let $S_i := S \cap V_i$. Let $F = \{ e \in H : e \in S \}$ and $F_{i,j} := F \cap H_{i,j}$. Our goal is to show $w(F)$ is approximately at least $1/8 - o(1)$. 

Given a vertex $v \in V_i$, let $C_v := \{ v \} \times C_i \subseteq U$ and 
\[
\alpha_v := \ell |V_i| \cdot \bigg( \sum_{v \in (S \cap C_v)} w(v) \bigg)
\]
be the normalized weight of $S$ in $C_v$. 
Given vertices $v_i \in V_i$ and $v_j \in V_j$ with $i < j$, let 
$D_{i,j} = \Pr_{(i', j') \sim \cald}[i = i', j = j']$ and 
\[
\beta_{v_i, v_j} := \frac{1}{D_{i,j}} \cdot |E_{i,j}| \cdot \bigg( \sum_{e \in (F \cap H(C_{v_i} \cup C_{v_j}))} w(e) \bigg)
\]
be the normalized weight of $F$ in $H(C_{v_i} \cup C_{v_j})$, where given $T \subseteq V$, $H(T)$ is defined as $\{ e \in H: e \subseteq T \}$. 
Note that all $\alpha_v, \beta_{v_i, v_j}$ are in $[0, 1]$. Furthermore, 
\[
\Ex_{i \in [\ell]} \Ex_{v \in V_i} [\alpha_v] = w(S),
\]
and 
\[
\Ex_{(i,j) \sim \cald} \Ex_{(v_i, v_j) \in E_{i,j}} [\beta_{v_i, v_j}] = w(F).
\]
Let $\alpha_i := \Ex_{v \in V_i} [\alpha_v]$ and $\beta_{i,j} := \Ex_{(v_i, v_j) \in E_{i,j}} [\beta_{v_i, v_j}]$.
Finally, for each full path $p = (v_1, \dots, v_{\ell})$, let $\alpha_{p, i} := \alpha_{v_i}$ and $\beta_{p, i,j} := \beta_{v_i, v_j}$. 

Call a triple $(p, i, j)$ {\em good} if $\beta_{p, i,j} < \alpha_{p, i} \alpha_{p, j}^2 - \eps$. 
The following lemma says that we are done if few triples are good. 

\begin{lemma}
If $\Pr_{p \in \calp, (i,j) \in \cald}[ (p, i, j)\mbox{ is good}] \leq \eps$, then $w(F) \geq 1/8 - 3(\sqrt{\eps} + 1/\ell)$. 
\label{lem:np-soundness-1}
\end{lemma}

Therefore, it remains to consider the case that the condition of Lemma~\ref{lem:np-soundness-1} does not hold. 
Then, there exists $i < j$ such that 
$\Pr_{p \in \calp}[ (p, i, j)\mbox{ is good}] \geq \eps.$
The following lemma, essentially from \hastad's analysis for Max 3-SAT~\cite{H01}, states that it cannot happen if there is no good assignment for the multilayered PCP instance $\calm$. 
For completeness, we reproduce a proof in Section~\ref{sec:hastad}

\begin{lemma}
Fix any $i < j$. 
If 
$\Pr_{p \in \calp}[ (p, i, j)\mbox{ is good}] \geq \eps$,
then there is an assignment for $\calm$ that satisfies at least an
$(\eps \delta^{4} / 2) \cdot (\eps/2 - 2\delta - 2\eta\delta^{-4})^2$
fraction of edges in $E_{i,j}$. 
\label{lem:np-soundness-2}
\end{lemma}

First take small enough $\eps > 0$ and large enough $\ell \in \NN$ will ensure $w(F) \geq 1/8 - 3(\sqrt{\eps} + 1/\ell)$ is almost at least $1/8$.
Choosing $\delta = \eps^2/4$, $\eta = \delta^6$ will ensure that the guarantee in Lemma~\ref{lem:np-soundness-2} is at least $\eps^c$ for some absolute constant $c \in \NN$. 
By taking $\tau$ in Theorem~\ref{thm:dgkr} smaller than that, we can ensure that $\Pr_{p \in \calp}[ (p, i, j)\mbox{ is good}] \leq \eps$ and $w(F) \geq 1/8 - 3(\sqrt{\eps} + 1/\ell)$ in the soundness case,
proving Theorem~\ref{thm:hvc}.

\subsection{Proof of Lemma~\ref{lem:np-soundness-1}}
\label{sec:hvc_proof_1}

\begin{proof}
Recall that for any $i < j$, sampling $p = (v_1, \dots, v_{\ell}) \in \calp$ and choosing $(v_i, v_j)$ is the same as sampling $(v_i, v_j) \in E_{i, j}$ uniformly at random. Therefore, 
\begin{align*}
w(F) & =
\Ex_{(i,j) \sim \cald} \Ex_{(v_i, v_j) \in E_{i,j}} [\beta_{v_i, v_j}] \\
& = \Ex_{(i,j) \sim \cald} \Ex_{p \in \calp} [\beta_{p, i, j}] \\
&= \Ex_{p \in \calp} \Ex_{(i,j) \sim \cald}  [\beta_{p, i, j}].
\end{align*}
Say $p$ is {\em atypical} 
$\Pr_{(i,j) \in \cald} [ (p, i, j)\mbox{ is good}] \geq \sqrt{\eps}$. 
Since $\Pr_{p \in \calp, (i,j) \in \cald}[ (p, i, j)\mbox{ is good}] \leq \eps$, 
\[
\Pr_{p \in \calp}[p \mbox{ is atypical}] \leq \sqrt{\eps}.
\]
Fix a typical $p$. Then 
\[
\Ex_{(i,j) \sim \cald}  [\beta_{p, i, j}] \geq 
\Ex_{(i,j) \sim \cald}  [\alpha_{i,p} \alpha_{j,p}^2 - \eps] - \sqrt{\eps},
\]
since we can apply the lower bound $\beta_{p,i,j}$ by $\alpha_{i,p}\alpha_{j,p}^2 - \eps$ with probability at least $1 - \sqrt{\eps}$ and $0$ otherwise. 

Now we analyze $\Ex_{(i,j) \sim \cald}  [\alpha_{i,p} \alpha_{j,p}^2]$. 
Recalling the definition of $\cald$ and applying Cauchy-Schwarz, 
\begin{equation}
\Ex_{(i,j) \sim \cald}  [\alpha_{i,p} \alpha_{j,p}^2] \geq 
\Ex_{i \sim \cald_I, j \in \{ i+1, \dots, \ell \} }  [\alpha_{i,p} \alpha_{j,p}^2] \geq
\Ex_{i \sim \cald_I, j,k \in \{ i+1, \dots, \ell \} }  [\alpha_{i,p} \alpha_{j,p} \alpha_{k, p}].
\label{eq:np-dist1}
\end{equation}
We compare the RHS of~\eqref{eq:np-dist1} to 
\begin{equation}
\big(\Ex_{i} [\alpha_{i,p}]\big)^3 = 
\Ex_{i, j,k \in [l] }  [\alpha_{i,p} \alpha_{j,p} \alpha_{k, p}].
\label{eq:np-dist2}
\end{equation}
If we fix $i < j < k$, 
the probability that the monomial $\alpha_{i,p} \alpha_{j,p} \alpha_{k,p}$ contributes to the expectation, after incorporating permutations between $i,j,k$, is 
$
6/\ell^3
$
in~\eqref{eq:np-dist2} and 
\[
\frac{(\ell-i)^2}{\ell(\ell-1)(2\ell-1)/6} \cdot \frac{2 }{(\ell - i)^2} = \frac{2 \cdot 6}{\ell(\ell-1)(2\ell - 1)}
\]
in~\eqref{eq:np-dist1}, which is greater than $6/\ell^3$. Since $\Pr[i=j \mbox{ or } j=k \mbox{ or } k=i] \leq \frac{3}{\ell}$ when we sample $i, j, k \in [\ell]$ independently, 
\[
\Ex_{i \sim \cald_I, j,k \in \{ i+1, \dots, \ell \} }  [\alpha_{i,p} \alpha_{j,p} \alpha_{k, p}]
\geq 
\Ex_{i, j,k \in [l] }  [\alpha_{i,p} \alpha_{j,p} \alpha_{k, p}] - \frac{3}{\ell} = \big(\Ex_{i} [\alpha_{i,p}]\big)^3 - \frac{3}{\ell}. 
\]
Therefore, if $p$ is typical, then 
\begin{equation*}
\Ex_{(i,j) \sim \cald}  [\beta_{p, i, j}] \geq \bigg( \Ex_{i \in [\ell]} [\alpha_i] \bigg)^3 - \frac{3}{\ell} - 2\sqrt{\eps}.
\end{equation*}
Using $\Ex_p \Ex_i [\alpha_{p, i}] \geq 1/2$ and Jensen's inequality, 
\begin{align*}
w(F) & =\Ex_{p \in \calp} \Ex_{(i,j) \sim \cald}  [\beta_{p, i, j}] \\ 
&\geq \Ex_{p \in \calp} \bigg[ \big( \Ex_{i \in [\ell]} [\alpha_i] \big)^3 \bigg] - \frac{3}{\ell} - 2\sqrt{\eps} - \Pr_{p \in \calp}[p \mbox{ is atypical}] \\
&\geq \bigg( \Ex_{p \in \calp} \Ex_{i \in [\ell]} [\alpha_i]  \bigg) ^3  - \frac{3}{\ell} - 3\sqrt{\eps} \\
&\geq \frac{1}{8}  - \frac{3}{\ell} - 3\sqrt{\eps}.
\end{align*}
\end{proof}

\subsection{Proof of Lemma~\ref{lem:np-soundness-2}}
\label{sec:hvc_proof_2}
\label{sec:hastad}
\begin{proof}
Fix $i < j$ given in the condition of the lemma. 
Call an edge $e = (v_i, v_j) \in E_{i,j}$ good if $\beta_{v_i, v_j} < \alpha_{v_i} \alpha_{v_j}^2 - \eps$. 
For $v_i \in V_i$, let $f_{v_i} : C_i \to \{ 0, 1 \}$ be the indicator function of $S \cap (\{ v_i \} \times C_i)$, and 
Let $v_j \in V_j$, let $g_{v_j} : C_j \to \{ 0, 1 \}$ be the indicator function of $S \cap (\{ v_j \} \times C_j)$. 

The path regularity and the promise of the lemma implies that at least an $\eps$ fraction of $e \in E_{i,j}$ is good. 
Fix such an edge $e = (v_i, v_j)$. For notational simplicity, let $\pi = \pi_e$, $L = \Sigma_i$, $R = \Sigma_j$,
$f := f_{v_i}$ and $g := g_{v_j}$. 

We use the standard notations in analysis of boolean functions. See~\cite{H01, Odonnell14} for references. 
For two functions $f_1, f_2 : C_i \to \R$, let $\langle f_1, f_2 \rangle := \Ex_{x \in \{ \pm 1 \}^{L}} [f_1(x) f_2(x)]$ be the inner product between $f_1$ and $f_2$. 
For $A \subseteq L$, let $\chi_A : C_i \to \R$ defined as $\chi_A(x) = \prod_{a \in A} x_a$.
It is well known that $\{ \chi_A \}_{A \subseteq L}$ forms an orthonormal basis, so that 
$f$ can be written as $\sum_{A \subseteq L} \wh{f}(A) \chi_A$ with $\wh{f}(A) = \langle f, \chi_A \rangle$. 
Define $\{ \chi_B \}_{B \in \subseteq R}$ similarly and write $g$ as $\sum_{B \subseteq R} \wh{g}(B) \chi_B$. 

Now note that $\alpha_i = \Ex_x[f(x)] = \wh{f}(\emptyset)$, 
$\alpha_j = \Ex_y[g(y)] = \wh{g}(\emptyset)$, and $\beta_{i,j} = \Ex_{x, y, z} [f(x) g(y) g(z)]$ where $x, y, z$ were jointly sampled the reduction given $e$. 
Expanding Fourier decompositions for $f$ and $g$, 
\begin{align}
\Ex_{x, y, z} [f(x) g(y) g(z)] 
&= \sum_{A \subseteq L, B \subseteq R, C \subseteq R} \wh{f}(A) \wh{g}(B) \wh{g}(C) \Ex[\chi_A(x) \chi_B(y) \chi_C(z)] \nonumber \\
&= \sum_{B \subseteq R, A \subseteq \pi(B)} \wh{f}(A) \wh{g}(B)^2 \Ex[\chi_A(x) \chi_B(y) \chi_B(z)] \nonumber \\
&=\sum_{B \subseteq R} \wh{g}(B)^2 \sum_{A \subseteq \pi(B)} \wh{f}(A) \Ex[\chi_A(x) \chi_B(y) \chi_B(z)].
\label{eq:fourier}
\end{align}
The second equality holds because if $b \in B \setminus C$, then $\chi_B(y)$ contains $y_b$ and it is independent from any other variable appearing in $\chi_A(x) \chi_B(y) \chi_C(z)$.
Similarly, the existence of $c \in C \setminus B$ or $a \in A \setminus (B \cup C)$ will make $\chi_A(x) \chi_B(y) \chi_C(z)$ vanish. 

Suppose $B \subseteq \pi^{-1}(a)$ for some $a \in L$. 
For each $b \in B$, $\Ex[y_b z_b] = -1$ if $x_a = 1$ and $(1 - 2\delta)$ otherwise, so 
\begin{equation*}
\Ex[\chi_B(y) \chi_B(z)] = \frac{1}{2} \big( (-1)^{|B|} + (1-2\delta)^{|B|} \big),
\end{equation*}
and 
\begin{equation*}
\Ex[x_a \chi_B(y) \chi_B(z)] = \frac{1}{2} \big( (-1)^{|B|} - (1-2\delta)^{|B|} \big).
\end{equation*}
Therefore, if we consider $\Ex[\chi_A(x) \chi_B(y) \chi_B(z)]$ for general $A \subseteq B$, 
letting $s_a := |B \cap \pi^{-1}(a)|$, 
$p_s := ( (-1)^{s} + (1-2\delta)^{s} ) /2$, 
$q_s := ( (-1)^{s} - (1-2\delta)^{s} ) /2$, 
it is equal to 
\begin{equation}
\bigg( \prod_{a \in A} q_{s_a} \bigg) \cdot
\bigg( \prod_{a \in \pi(B) \setminus A} p_{s_a} \bigg).
\end{equation}
Note that $p_s^2 + q_s^2 \leq 1 - \delta$ for any $s \geq 1$. Then for fixed $B$, 
\begin{align*}
\sum_{A \subseteq \pi(B)} \big( \Ex[\chi_A(x) \chi_B(y) \chi_B(z)] \big)^2
& \leq \sum_{A \subseteq \pi(B)} \Ex[( \chi_A(x) \chi_B(y) \chi_B(z))^2]  \\
& = \sum_{A \subseteq \pi(B)} \bigg( \prod_{a \in A} q_{s_a}^2 \prod_{a \in \pi(B) \setminus A} p_{s_a}^2 \bigg) \\
& = \prod_{a \in \pi(B)} (p_{s_a}^2 + q_{s_a}^2) \\
& \leq (1-\delta)^{|\pi(B)|}.
\end{align*}

Finally, we analyze~\eqref{eq:fourier}. 
When $B = \emptyset$, we get $\wh{f}(\emptyset)\wh{g}(\emptyset)^2 = \alpha_i \alpha_j^2$. 
Say $B$ {\em big} if $|B| > \delta^{-2}$ and small if $1 \leq |B| \leq \delta^{-2}$. 
Fix large $B$ and $v \in V_j$, and consider a random edge $(u, v) \in E_{i,j}$.
Since $\calm$ is $\eta$-smooth, the probability that $|\pi(B)| \geq \delta^{-2}$ is at least $1 - \eta \delta^{-4}$, so using $(1-\delta)^{1 / 2\delta^2} \leq \delta$, 
\begin{align*}
\Ex_{(u, v) \in E_{i,j}} (1-\delta)^{|\pi(\beta)| / 2} \leq \delta + \eta \delta^{-4}. 
\end{align*}
Therefore, for any fixed $v \in V_j$, we can bound~\eqref{eq:fourier} for big $B$ as: \allowdisplaybreaks

\begin{align*}
&\Ex_{(u, v) \in E_{i,j}} \bigg[ \bigg| \sum_{B \mbox{ big}} \wh{g}(B)^2 \sum_{A \subseteq \pi(B)} \wh{f}(A) \Ex[\chi_A(x) \chi_B(y) \chi_B(z)] \bigg| \bigg]  \\
\leq \quad & \Ex_{(u, v) \in E_{i,j}} \bigg[ \sum_{B \mbox{ big}} \wh{g}(B)^2 \bigg( \sum_{A \subseteq \pi(B)} \wh{f}(A)^2  \bigg)^{1/2} \bigg( \sum_{A \subseteq \pi(B)} \Ex[\chi_A(x) \chi_B(y) \chi_B(z)]^2 \bigg)^{1/2} \bigg] \\
\leq \quad &\Ex_{(u, v) \in E_{i,j}} \bigg[ \sum_{B \mbox{ big}} \wh{g}(B)^2 (1 - \delta)^{|\pi(\beta)|/2} \bigg] \leq \delta + \eta \delta^{-4} . 
\end{align*}
Similarly, we can bound~\eqref{eq:fourier} for small $B$ and $A = \emptyset$. 
With probability at least $1 - \eta |\pi(B)|^2 \geq 1 - \eta \delta^{-4}$ we have $|\pi(\beta)| = |\beta|$, and if this happens, 
$\Ex[|\chi_B(y) \chi_B(z)|] \leq |p_1| = \delta$. Therefore, 
\begin{align*}
&\Ex_{(u, v) \in E_{i,j}} \bigg[ \bigg| \sum_{B \mbox{ small}} \wh{g}(B)^2 \wh{f}(\emptyset) \Ex[\chi_B(y) \chi_B(z)] \bigg| \bigg]  \\
&\Ex_{(u, v) \in E_{i,j}} \bigg[ \bigg| \sum_{B \mbox{ small}} \Ex[\chi_B(y) \chi_B(z)] \bigg| \bigg] \\ 
& \eta \delta^{-4} + \delta.
\end{align*}
Finally, for small $B$ and $\emptyset \subsetneq A \subseteq \pi(B)$, we bound~\eqref{eq:fourier} as 
\begin{align*}
& \bigg| \sum_{B \mbox{ small}} \wh{g}(B)^2 \sum_{\emptyset \subsetneq A \subseteq \pi(B)} \wh{f}(A) \Ex[\chi_A(x) \chi_B(y) \chi_B(z)] \bigg| \\
\leq \quad &
\bigg( \sum_{B \mbox{ small}} \wh{g}(B)^2 \sum_{\emptyset \subsetneq A \subseteq \pi(B)} \wh{f}(A)^2 \bigg)^{1/2}
\bigg( \sum_{B \mbox{ small}} \wh{g}(B)^2 \sum_{\emptyset \subsetneq A \subseteq \pi(B)} \Ex[\chi_A(x) \chi_B(y) \chi_B(z)]^2 \bigg)^{1/2} \\
\leq \quad & \bigg( \sum_{B \mbox{ small}} \wh{g}(B)^2 \sum_{\emptyset \subsetneq A \subseteq \pi(B)} \wh{f}(A)^2 \bigg)^{1/2}.
\end{align*}

Since at least an $\eps$ fraction of $e \in E_{i,j}$ is good, at least an $\eps/2$ fraction of $v \in V_j$ satisfies that at least an $\eps/2$ fraction of $(u, v) \in E_{i, j}$ is good. 
Call such $v$ {\em good}. If $v$ is good, 
\begin{align*}
& \eps /2 \leq \Ex_{(u, v) \in E_{i,j}} \bigg[ \bigg| \beta_{u,v} - \alpha_u \alpha_v^2 \bigg| \bigg] 
\leq 2\delta + 2\eta \delta^{-4} + 
\Ex_{(u, v) \in E_{i,j}} \bigg[ \bigg( \sum_{B \mbox{ small}} \wh{g}(B)^2 \sum_{\emptyset \subsetneq A \subseteq \pi(B)} \wh{f}(A)^2 \bigg)^{1/2} \bigg].
\end{align*}

Consider the randomized assignment where all $v \in V_j$ first chooses a set $B \subseteq R$ with probability $\wh{g}(B)^2$ and gets random $b \in B$. 
(Similarly, $u \in V_i$ chooses a set $A \subseteq L$ with probability $\wh{f}(A)^2$ and gets random $a \in A$.) 
Since the sum of squared Fourier coefficients is at most $1$ for every $f$ and $g$, this is a well-defined strategy. For good $v$, it will satisfy at least a
\[
\frac{(\eps/2 - 2\delta - 2\eta\delta^{-4})^2}{|A||B|} \leq 
\delta^{4} (\eps/2 - 2\delta - 2\eta\delta^{-4})^2 
\]
fraction of constraints incident on $v$ in expectation. Therefore, there is an assignment between that satisfies at least an
$(\eps \delta^{4}/2) \cdot (\eps/2 - 2\delta - 2\eta\delta^{-4})^2$ fraction of edges in $E_{i, j}$. 
\end{proof}

\subsection{Make instances dense}
\label{subsec:dense}
In this section, we show how to convert hard instances to ensure $|H| = \omega(|V|^2)$ while preserving hardness, finishing the proof of Theorem~\ref{thm:hvc}. 
From the previous discussion, given an edge-weighted $3$-hypergraph $\calh = (V, H)$ with $n = |V|$, $m = |H|$, and $k = n/2$
(without loss of generality, assume that the sum of weights is $1$), 
it is NP-hard to distinguish whether (1) there exists $S \subseteq V$ that intersects every hyperedge or (2) 
every $S \subseteq V$ with $|S| \leq n/2$ intersects hyperedges of weight at most $(7/8+\eps)$, for any constant $\eps > 0$. 

Let $b = \max(n, m)^{\beta}, c = b^{2.5}$ where $\beta$ is a constant chosen later. Our reduction creates a new hypergraph $\calh' = (V', H')$ where
\begin{itemize}
\item $V' = V \times [b]$. 
\item For each hyperedge $(u, v, w) \in H$ with weight $w(u,v,w)$, for $\lfloor c \cdot w(u,v,w) \rfloor $ times independently, 
\begin{itemize}
\item Sample $x, y, z \in [b]$ independently. 
\item Add hyperedge $((u,x), (v,y), (w,z))$ to $H'$. 
\end{itemize}
\item If any hyperedge is added more than once, delete all occurrences of the hyperedge. 
\end{itemize}
By the last step, $\calh'$ is simple. The number of hyperedges added before the last step is at least 
\[
\sum_{e \in H} \lfloor c \cdot w(e) \rfloor
\geq \big( \sum_{e \in H} c \cdot w(e) \big) - m = c - m.
\]
Fix $(u, v, w) \in H$ and let $(x, y, z)$ be the triple sampled in the $i$th iteration for $(u, v, w)$ such that $((u,x), (v,y), (w,z))$ is added to $H'$. 
The probability that the same triple $(x, y, z)$ is chosen in another iteration so that $((u,x), (v,y), (w,z))$ is deleted in the last step is at most $c / b^3$. 
Therefore, the expected number of hyperedges deleted in the last step of the reduction is at most $c \cdot (c/b^3)$, and 
the with probability at least $0.9$, it is at most $c \cdot (10c / b^3)$. 

\paragraph{Completeness.} If $S \subseteq V$ intersects every hyperedge in $H$ with nonzero weight, $S' = S \times [b]$ does the same for $H'$. 

\paragraph{Soundness.} For soundness, we upper bound the number of hyperedges before the last deletion step intersected by $S'$ with $|S'| \leq |V'| / 2$, because this is only an overestimate.  
Fix $e = (u, v, w) \in E$ and 
consider the hyperedges created when considering $e$ and let $c_e = \lfloor c \cdot w(e) \rfloor$ be the number of them. 
Fix $X, Y, Z \in [b]$ and $\alpha_u = |X|/b$, $\alpha_v = |Y|/b$, $\alpha_w = |Z|/b$. The expected number of hyperedges not intersecting $(\{ x \} \times X) \cup (\{ y \} \times Y) \cup (\{ z \} \times Z)$ is exactly 
$c_e (1-\alpha_u)(1-\alpha_v)(1-\alpha_w)$. 
By Hoeffding's bound, the probability that it is less than the expected value minus $t$ is at most $\exp(-2t^2 / c_e)$. 
By union bound, the probability that this happens for any choice of $u, v, w, X, Y, Z$ is at most 
\[
m \cdot 2^{3b} \cdot \exp(-2t^2 / c_e) \leq m \cdot \exp(3b - 2t^2 / c).
\]
By taking $t = b^{1.8}$ ensures $2t^2 / c = 2b^{1.1}$ so that the above probability is at most $m \exp(-b^{1.1})$. Since $b$ will be greater than $m$, this probability is $o(1)$. 
Therefore, with probability $1 - o(1)$, for any $(u, v, w) \in H$ and $X, Y, Z \subseteq [b]$, 
the number of hyperedges created from $(u, v, w)$ not intersecting $(\{ x \} \times X) \cup (\{ y \} \times Y) \cup (\{ z \} \times Z)$ is at least 
\[
c_e (1-\alpha_u)(1-\alpha_v)(1-\alpha_w) - t.
\]
Then for any $S' \subseteq V$ with $|S'| = |V'| / 2$, let $\alpha_v := | (\{ v \} \times [b]) \cap S' | / b$ for $v \in V$, so that $\Ex_v[\alpha_v] =1/2$. 
Then the total number of edges not intersecting $S'$ is at least 
\begin{align*}
&\bigg( \sum_{e = (u, v, w) \in H} c_e (1-\alpha_u)(1-\alpha_v)(1-\alpha_w) \bigg) - mt\\
&\geq 
c \cdot \bigg( \sum_{e = (u, v, w) \in H} w_e (1-\alpha_u)(1-\alpha_v)(1-\alpha_w) \bigg) - m(t + 1).
\end{align*}
Note that the first term of the RHS is exactly $c$ times the expected weight of edges of $H$ not intersecting $S$, where $S \subseteq V$ is a random subset that includes $v \in V$ with probability $\alpha_v$ independently. By invoking the soundness condition for $\calh = (V, H)$, the RHS is at least $c(1/8 - \eps) - m(t+1)$. 

\paragraph{Finishing up.}
Therefore, with probability at least $0.9 - o(1)$, $\calh'$ has at least $c - m - 10c^2 / b^3$ hyperedges. 
In the completeness case, there exists $S$ with $|S| \leq |V'| / 2$ that intersects every hyperedge,
and in the soundness case, every $S$ with $|S| \leq |V'| / 2$ does not intersect at least $c(1/8 - \eps) - m(t+1)$ hyperedges. 
Recall the parameter setting $c = b^{2.5}$, $t = b^{1.8}$ and $b = \max(n, m)^{\beta}$. 
Setting $\beta \geq 2$ will ensure $m(t+1) = o(c)$, so that $\calh'$ has $(1-o(1))c$ hyperedges and 
the gap is preserved to be $7/8 + \eps + o(1)$. The number of vertices $|V'| = nb \leq b^{1 + 1/\beta}$ and the number of hyperedges 
$|H'| \geq (1 - o(1))c \geq \Omega(b^{2.5})$, so setting $\beta = 5$ will ensure that $|H'| = \omega(|V'|^{2})$.

\section{Open Problems}\label{sec:open}

In this section, we list some open problems   related to hardness of approximation of clustering objectives. In this regard, the most important and also immediate question is whether \JCH is true?

\begin{open}\label{open:JCH}
Is the Johnson Coverage Hypothesis true?
\end{open}

Another important question is whether there is a `black-box' way to ensure that the hard instance of \JCH has large number of clusters. We point the reader to 
Section~\ref{subsec:dense} which seems to be a first step in this direction.

\begin{open}\label{open:JCHD}
Does \JCH imply \JCHD?
\end{open}

Next, we move to discuss open questions with a more combinatorial-geometric flavor. We showed in Lemma~\ref{lem:cd2} that $g_2(J(q,t+1,t))\ge \sqrt{1+\frac{1}{\sqrt{t^2+t}-t}}$. As $t$ increases, the value of $\sqrt{1+\frac{1}{\sqrt{t^2+t}-t}}$ converges to $\sqrt{3}$. However, the naive upper bound from triangle inequality (Proposition~\ref{prop:ubcd}) states that  
$\gamma_2\le 3$. For small values of $t$ we can indeed obtain an improved value: it is easy to see that $g_2(J(3,2,1))= 2$   by placing the six points on the vertices of a regular hexagon in the plane. On the other hand, we suspect that $\gamma_2\le 2$, and confirming such a claim would also be interesting.    

\begin{open}\label{open:ell2}
Is  $\gamma_2\ge 2$?\end{open}

We provided a few lower bounds on $\gamma_p$ in Section~\ref{sec:gadget}, but we are still very far from having good bounds on it. 

\begin{open}\label{open:ellp}
Is there a closed form expression for $\gamma_p$ (in terms of $p$)?\end{open}

We now shift our attention to understanding various clustering objectives. First, we highlight that there is no inapproximability results for \kmed and \kmean in $\ell_\infty$-metric in the continuous case in $O(\log n)$-dimensions.  The main obstacle that  we are not able to overcome is  that  many natural embedding techniques create fake centers in the $\ell_\infty$-metric. We emphasize the requirement of $O(\log n)$-dimensions because  in higher  dimensions (i.e., $\poly(n)$ dimensions), we were recently able to prove very strong inapproximability for \kmean and \kmed without candidate centers in $\ell_\infty$-metric \cite{CKL21}.  

\begin{open}\label{open:ellinf}
What is the hardness of approximation for \kmean and \kmed in the $\ell_\infty$-metric in $O(\log n)$ dimensions?
\end{open}

In \cite{CKL21} we  highlighted that there are inherent differences between the continuous and discrete cases for clustering problems and maybe basing hardness of clustering problems in the continuous case on \JCH might not lead to a tight understanding. Elaborating, by starting from coloring problems (instead of covering problems), we obtained strong inapproximability results for clustering problems in the continuous case in the $\ell_\infty$-metric. It is natural to ask if this approach can be extended to other $\ell_p$-metrics.

\begin{open}\label{open:contell2}
Can we show inapproximability of \kmean in the continuous case  to a factor more than  $1+\nicefrac{1}{e}$?
\end{open}

Apart from \kmean and \kmed, another clustering objective of interest is \minsum (see Section~\ref{sec:prelims} for the definition). However, for \minsum in $\ell_p$-metrics (for finite $p$), only \APX-Hardness is known, leaving it open to prove stronger inapproximability results. 

\begin{open}\label{open:ellpminsum}
Does \JCH (or \JCHD) imply any non-trivial inapproximability results for \minsum in $\ell_p$-metrics?
\end{open}

Finally, we discuss a more technical challenge. Our analysis of continuous case of \kmed in the Euclidean metric in  Theorem~\ref{thm:kmedl2} is not tight and we raise the following question.

\begin{open}\label{open:contell2kmed}
Assuming \JCHD, can we show that \kmed in Euclidean metric is hard to approximate to a factor less than $1+\frac{\sqrt{2}-1}{e}$?
\end{open}

%
%
%
%

\subsection*{Acknowledgements}
We are truly grateful to Pasin Manurangsi for various detailed discussions that inspired many of the results in this paper.

Ce projet a b\'en\'efici\'e d'une aide de l'\'Etat g\'er\'ee
par l'Agence Nationale de la Recherche au titre du Programme
Appel à projets générique JCJC 2018 portant la r\'ef\'erence
suivante : ANR-18-CE40-0004-01.
Karthik C.\ S.\ was supported by Irit Dinur's ERC-CoG grant 772839, the  Israel Science Foundation (grant number 552/16), the Len Blavatnik, the Blavatnik Family foundation, Subhash Khot's Simons Investigator Award, and by a grant from the Simons Foundation, Grant Number 825876, Awardee Thu D. Nguyen. 
Euiwoong Lee was supported in part by the Simons Collaboration on Algorithms and Geometry.

\bibliographystyle{alpha}
\bibliography{references}

\appendix

\section{Obstacles that need to be Overcome to prove Johnson Coverage Hypothesis}\label{sec:labelcover}

In this section, we consider a few natural approaches to proving \JCH, and demonstrate central obstacles to proceed using this approach. 

It is well known that \mkc is hard to approximate to factor beyond $1-\nicefrac{1}{e}$ and one may wonder if there are any ``simple'' gap preserving reductions from \mkc to Johnson Coverage problem. In the below theorem, we show that if there is such a reduction then there should be a significant blowup in size of the witness.  In fact our result is stronger, as we show that even to reduce gap \mkc to exact Johnson Coverage problem (i.e., at the cost of losing the all the gap in \mkc), we still need to blow up in witness size.  

\begin{theorem}\label{thm:fpt}
Let $\alpha,\delta\ge 0$. Let $z$ be some fixed constant. 
Suppose there is an algorithm $\mathcal{A}$ that on given as input a  \mkc instance $(\mathcal{U},\mathcal{S},k)$, outputs a $(\alpha,z)$-Johnson Coverage instance $([n],E,k')$ such that the following holds. 
\begin{description}
\item[Size:]  $|\mathcal{U}|=|\mathcal{S}|^{O(1)}=n^{O(1)}$ and $k'=F(k)$ for some computable function $F$.
\item[Completeness:] If there are $k$ sets in $\mathcal{S}$ that cover $\mathcal{U}$ then there exists $\C:=\{S_1,\ldots ,S_{k'}\}\subseteq  \binom{[n]}{z-1}$ such that 
$$\con(\C):=\underset{i\in[k']}{\cup} \con(S_i,E)=E.$$
\item[Soundness:]  If no $k$ sets in $\mathcal{S}$ cover $(1-\delta)$ fraction of $\mathcal{U}$ then for every $\C:=\{S_1,\ldots ,S_{k'}\}\subseteq \binom{[n]}{z-1}$ we have ${\left|\con(\C)\right|}\le \alpha\cdot\left|E\right|$.
\item[Running Time:]  $\mathcal{A}$ runs in time $T(k)\cdot \poly(n)$ for computable function $T$.
\end{description}
Then the following consequences hold:
\begin{itemize}
\item If $\delta=0$ then W[2]=FPT.
\item If $\delta\le \nicefrac{1}{k}$ then W[1]=FPT.
\item If $\delta\le \nicefrac{1}{e}-\varepsilon$ for some constant $\varepsilon>0$ then \gapeth is false.
\end{itemize} 
\end{theorem}

In order to understand the consequences, we remark the following:
\begin{remark}Note that $W[2]\neq FPT$ is a weaker assumption than $W[1]\neq FPT$, which is in turn a weaker assumption than $\mathsf{ETH}$ and consequently \gapeth.
\end{remark}

The proof of Theorem~\ref{thm:fpt} follows from Lemma~\ref{lem:JCHFPT}, Theorem~\ref{thm:FPTmkc}, Theorem~\ref{thm:gapeth}, and a result of \cite{DF95}.

\begin{lemma}\label{lem:JCHFPT}
There is an algorithm running in time $z^k\cdot n^{O(1)}$ that can decide any $(0,z)$-Johnson Coverage instance.
\end{lemma}
\begin{proof} We essentially follow the same idea as the FPT algorithm for vertex cover. 
Let $([n],E,k)$ be a $(0,z)$-Johnson Coverage instance. Pick an arbitrary set $S\in E$. There are $z$ possible subsets in $\binom{[n]}{z-1}$, say $T_1,\ldots T_z$, that can cover $S$. We branch and consider all $z$ possibilities. Suppose we branch and pick $T_i$. Then we remove all subsets in $E$ that are covered by $T_i$, and repeat by picking another arbitrary subset  $S'$ in $E$. We stop the algorithm after the branching tree is of height $k$.  If there are $k$ subsets in $\binom{[n]}{z-1}$ that cover all subsets in $E$, then in one of the branching, we would have found it. The size of the branching tree is at most $z^k$. and at each step, updating $E$ can be done in linear time. 
\end{proof}

\begin{theorem}[Essentially \cite{KLM19}]\label{thm:FPTmkc}
Given an instance $(\mathcal{U},\mathcal{S},k)$ of \mkc, it is W[1]-Hard (when parameterized by $k$) to distinguish between the following two cases:
 \begin{description}
\item[Completeness:] There are $k$ sets in $\mathcal{S}$ that cover $\mathcal{U}$.
\item[Soundness:]  No $k$ sets in $\mathcal{S}$ cover $\left(1-\frac{1}{k}\right)$ fraction of $\mathcal{U}$.
\end{description}
\end{theorem}
\begin{proof}[Proof Sketch]
We use the notation and terminology of MaxCover problem given in \cite{KLM19}.  Our starting point is the following result.
\begin{theorem}[\cite{KLM19}]
For every $\varepsilon>0$, given an instance $\Gamma$ of MaxCover as input, it is $W[1]$-Hard to distinguish between the following two cases:
 \begin{description}
\item[Completeness:] MaxCover($\Gamma$) = 1.
\item[Soundness:]  MaxCover($\Gamma$) $\le\ \varepsilon$.
\end{description}
Moreover, this holds even when the size of each left super node is a constant only depending on $k$ and $\varepsilon$.
\end{theorem}
Then we can simply apply Feige's \cite{F98} proof framework to conclude the corollary; details follow. Let the input instance to Maxcover be $\Gamma:=(U:=U_1\cup \cdots \cup U_q,W:=W_1\cup \cdots \cup W_k, E)$. We build an instance $(\mathcal{U},\mathcal{S},k)$ of \mkc as follows:
$$
\mathcal{U}:=\{(i,f)\mid i\in [t], f:U_i\to [k]\}, \mathcal{S}:=\{S_{j,w}\mid j\in[k], w\in W_j\}, \text{ and }
$$
$$
(i,f)\in S_{j,w} \Leftrightarrow \exists u\in U_i \text{ such that }f(u)=j\text{ and }(u,w)\in E.
$$

Its easy to see that a labeling of $W$ corresponds to a $k$-tuple of sets in $\mathcal{S}$. In the completeness case, a labeling of $W$ which yields MaxCover$(\Gamma)=1$ also corresponds to $k$ sets in $\mathcal{S}$ that cover all of $\mathcal{U}$. In the soundness case, we claim that there are at least $(1-\varepsilon)\cdot q$ many universe elements in $\mathcal{U}$ that are not covered by any $k$ sets in $\mathcal{S}$. Note that $|\mathcal{U}|=\sum_{i\in [q]}k^{|U_i|}=O_{k,\varepsilon}(q)$.  With the right choice of $\varepsilon$, the theorem statement follows.
\end{proof}

\begin{theorem}[\cite{Cohen-AddadG0LL19,M20}]\label{thm:gapeth}
Assuming \gapeth, for every $\varepsilon>0$, there is no algorithm running in time $n^{o(k)}$ which given an instance $([n],\mathcal{S},k)$ of \mkc, can distinguish between the following two cases:
 \begin{description}
\item[Completeness:] There are $k$ sets in $\mathcal{S}$ that cover $[n]$.
\item[Soundness:]  No $k$ sets in $\mathcal{S}$ cover $\left(1-\frac{1}{e}+\varepsilon\right)$ fraction of $[n]$.
\end{description}
\end{theorem}

\begin{proof}[Proof of Theorem~\ref{thm:fpt}]
Suppose there is an algorithm $\mathcal{A}$ as claimed in the theorem statement. Let  $(\mathcal{U},\mathcal{S},k)$ be a  \mkc instance. We run $\mathcal{A}$ to obtain  a $(\alpha,z)$-Johnson Coverage instance $([n],E,k')$. We then run the algorithm in Lemma~\ref{lem:JCHFPT} to distinguish if $([n],E,k')$ is part of the completeness case or the soundness case. Thus we obtained an algorithm for deciding $(\mathcal{U},\mathcal{S},k)$ that runs in time $T'(k)\cdot \poly(|\mathcal{U}|)$ for some computable function $T'$. Depending on the value of $\delta$, this contradicts either W[2]-Hardness of \mkc \cite{DF95}, Theorem~\ref{thm:FPTmkc}, or Theorem~\ref{thm:gapeth}.
\end{proof}

It's worth noting that Theorem~\ref{thm:fpt} only holds for constant $z$. In fact, if  it suffices to only show \gapeth is false as a consequence in Theorem~\ref{thm:fpt} then we can allow $z$ to be $n^{o(1)}$. On the other hand, we could  completely get away with any kind of restriction on $z$ by setting $\delta>0$ in Theorem~\ref{thm:fpt} and invoking the result in \cite{badanidiyuru2012approximating} instead of Lemma~\ref{lem:JCHFPT}.

Next, we show that in order to prove \JCH, we would require to prove \NP-hardness of a plausibly highly structured Label Cover instance. 
We first recaptiluate here proof of the $(1-\nicefrac{1}{e})$-factor inapproximability of \mkc  shown by Feige \cite{F98}. We present the proof outline below in terms of label cover (as in \cite{M15,DS14}) instead of multi-prover proof systems (as in \cite{LY94,F98}).

We formally define the label cover problem and state its hardness of approximation result that follows from the application of the parallel repetition theorem \cite{R98,DS14} to the PCP theorem \cite{AS98,ALMSS98}. Below we state a restricted bounded degree and bounded alphabet size version of gap label cover problem.

\begin{definition}[Label Cover problem\footnote{The label cover problem as defined here is known in literature as the label cover problem with projection property or as the projection game problem, but we drop the word `projection' here for brevity.}]
Let $\varepsilon>0$, $d,\alpha\in\mathbb{N}$. Let $\Sigma_U,\Sigma_V$ be two finite sets. The input to a $(\eps,d,\alpha)$-label cover problem $\Pi$ is a bipartite graph $G(U\cup V,E)$ and a set of projection functions $\pi=\{\pi_e:\Sigma_U\to \Sigma_V\mid e\in E \}$ such that the following holds:
\begin{itemize}
\item $|\Sigma_U|,|\Sigma_V|\le \alpha$.
\item for all $u\in U\cup V$, we have degree of $u$ is at most $d$.
\end{itemize} 
For every assignment $\sigma:=(\sigma_U:U\to\Sigma_U,\sigma_V:V\to\Sigma_V)$ to $\Pi$, we define $\sat(\Pi,\sigma)$ as follows:
$$
\sat(\Pi,\sigma):=\underset{e:=(u,v)\sim E}{\mathbb{E}}[\pi_{e}(\sigma_U(u))=\sigma_V(v)].
$$ 
The goal of the $(\eps,d,
\alpha)$-label cover problem is to distinguish between the following two cases.
\begin{itemize}
\item \textbf{\emph{Completeness}}: There exists an assignment $\sigma$ to $\Pi$ such that $\sat(\Pi,\sigma)=1$.
\item \textbf{\emph{Soundness}}: For every assignment $\sigma$ to $\Pi$ we have that $\sat(\Pi,\sigma)\le \varepsilon$.
\end{itemize}
\end{definition}

An immediate consequence of the PCP theorem is that it is \NP-hard to decide an instance $\Pi(G,\pi)$  of $(1-\eps,d,\alpha)$-label cover problem for some constants $\eps>0, d,\alpha\in\mathbb{N}$. By applying the parallel repetition theorem to the gap instances arising from the PCP theorem, and followed by a strengthening of the soundness guarantee due to Moshkovitz \cite{M15} we get the following.

\begin{theorem}[Bounded Label Cover Inapproximability \cite{AS98,ALMSS98,R98,M15}]\label{thm:labcov}
For every constant $\varepsilon>0$, there exist constants $d:=d(\varepsilon)\in\mathbb{N}$ and $\alpha:=\alpha(\varepsilon)\in\mathbb{N}$ such that it is \NP-hard to decide an instance $\Pi(G,\pi)$  of $(\eps,d,\alpha)$-label cover problem. Moreover, this holds even for the following soundness guarantee: for every assignment $\sigma_U:U\to\Sigma_U$ we have the following:
$$
\wsat(\sigma_U):=\frac{|\{v\in V\mid \exists u_1\neq u_2\in N(v), \pi_{u_1,v}(\sigma(u_1))=\pi_{u_2,v}(\sigma(u_2))\}|}{|V|}\le \varepsilon
$$
\end{theorem}

We define a gadget relevant to \mkc:

\begin{definition}
Let $d,q\in\mathbb{N}$ and $S\subseteq [d]^{q}$. We say that $S$ is $(d,q)$-resistant if the following holds: for every $H_1,\cdots, H_d$ be $d$ axis-parallel hyperplanes such that no two are mutually parallel, there is a point in $S$ not contained in any of the hyperplanes. 
\end{definition}

\begin{theorem}[Essentially Feige~\cite{F98}]\label{thm:maxvc1}
 Let $\Pi(G,\pi)$ be a hard instance of $(\eps,d,\alpha)$-label cover problem as in Theorem~\ref{thm:labcov} and let $S$ be a $(d,|\Sigma_V|)$-resistant set.  Consider the following set-system $(\mathcal{U},\mathcal{S})$:
 $$
\mathcal{U}:=\{(v,s)\mid v \in V, s\in S\}, \mathcal{S}:=\{S_{u,a}\mid u\in U, a\in \Sigma_U\}, \text{ and }
$$
$$
(v,s)\in S_{u,a} \Leftrightarrow  (u,v)\in E\text{ and }s_{\pi_{(u,v)}(a)}=v.
$$
Then, it is \NP-Hard to distinguish between the following:
\begin{description}
\item[Completeness:] There are $|U|$ sets in $\mathcal{S}$ that cover $\mathcal{U}$.
\item[Soundness:]  No $|U|$ sets in $\mathcal{S}$ cover $\left(1-\delta\right)$ fraction of $\mathcal{U}$ for some positive constant $\delta$ only depending on $\varepsilon$, $|\Sigma_U|$, $|\Sigma_V|$,  and $d$.
\end{description}
\end{theorem}

In fact by using $S$ to be the entire space $[d]^{|\Sigma_V|}$ and with a more tighther analysis, Feige showed the tight $1-1/e$ inapproximability of \mkc. He even noted that a random subset $S$ of small cardinality (roughtly $\tilde{O}(d|\Sigma_V|)$) suffices. We show below that there is no such subset $S$ for which we could use Feige's framework to obtain \JCH.

\begin{theorem}\label{thm:jchgadget}
 Let $\Pi(G,\pi)$ be a hard instance of $(\eps,d,\alpha)$-label cover problem as in Theorem~\ref{thm:labcov} and let $S$ be a $(d,|\Sigma_V|)$-resistant set.  Consider the following set-system $(\mathcal{U},\mathcal{S})$:
 $$
\mathcal{U}:=\{(v,s)\mid v \in V, s\in S\}, \mathcal{S}:=\{S_{u,a}\mid u\in U, a\in \Sigma_U\}, \text{ and }
$$
$$
(v,s)\in S_{u,a} \Leftrightarrow  (u,v)\in E\text{ and }s_{\pi_{(u,v)}(a)}=v.
$$
Suppose $(\mathcal{U},\mathcal{S})$ is an instance of $(0,z)$-Johnson Coverage problem (for some constant $z$) then \gapeth is false.
\end{theorem}

The proof follows by noting the following lower bound on parameterized label cover problem and Lemma~\ref{lem:JCHFPT}.

\begin{theorem}[Parameterized Label Cover Inapproximability with total disagreement \cite{M20}]\label{thm:fptlabcovtotal}
Assuming \gapeth, for every constant $\varepsilon>0$, there exist constants $d:=d(\varepsilon)\in\mathbb{N}$ and $\alpha:=\alpha(\varepsilon)\in\mathbb{N}$ such that no algorithm running in time $\alpha^{o(k)}$ given as input a instance $\Pi(G,\pi)$  of $(\eps,d,\alpha)$-label cover problem (where $|U|=k$), can distinguish between the following two cases.
\begin{itemize}
\item \textbf{\emph{Completeness}}: There exists an assignment $\sigma$ to $\Pi$ such that $\sat(\Pi,\sigma)=1$.
\item \textbf{\emph{Soundness}}: For every assignment $\sigma_U:U\to\Sigma_U$ we have $
\wsat(\sigma_U)\le \varepsilon.$
\end{itemize}
\end{theorem}

\end{document}